\documentclass[letterpaper, 11pt]{article}

\usepackage{amsmath,amssymb,amsfonts,mathrsfs,amsthm}
\usepackage[table,svgnames]{xcolor}
\usepackage{fullpage}

\usepackage{enumerate}
\usepackage{setspace}
\usepackage{lastpage}
\usepackage{extramarks}
\usepackage{algorithm}
\usepackage{algorithmic}
\usepackage{soul,color}
\usepackage{graphicx,float,wrapfig}
\usepackage{xspace}
\usepackage{framed}
\usepackage{verbatim}

\usepackage{enumitem}
\setenumerate{noitemsep}
\newtheorem{theorem}{Theorem}[section]

\newtheorem{claim}[theorem]{Claim}
\newtheorem{proposition}[theorem]{Proposition}

\newtheorem{lemma}[theorem]{Lemma}

\newtheorem{fact}[theorem]{Fact}

\theoremstyle{definition}
\newtheorem{definition}[theorem]{Definition}
\newcommand{\E}{\mathbb{E}}
\newcommand{\T}{\mathcal{T}}

\newcommand{\eps}{\varepsilon}

\newcommand{\EDTC}{EDTC\xspace}
\newcommand{\PEDTC}{PEDTC\xspace}

\newcommand{\uniform}{{\underset{U}{\gets}}}

\theoremstyle{plain}

\usepackage{tikz}
\usetikzlibrary{shapes,calc,arrows,petri}

\usetikzlibrary{calc,decorations.pathmorphing,patterns}
\pgfdeclaredecoration{penciline}{initial}{
    \state{initial}[width=+\pgfdecoratedinputsegmentremainingdistance,
    auto corner on length=1mm,]{
        \pgfpathcurveto%
        {
            \pgfqpoint{\pgfdecoratedinputsegmentremainingdistance}
                      {\pgfdecorationsegmentamplitude}
        }
        {
        \pgfmathrand
        \pgfpointadd{\pgfqpoint{\pgfdecoratedinputsegmentremainingdistance}{0pt}}
                    {\pgfqpoint{-\pgfdecorationsegmentaspect
                     \pgfdecoratedinputsegmentremainingdistance}%
                               {\pgfmathresult\pgfdecorationsegmentamplitude}
                    }
        }
        {
        \pgfpointadd{\pgfpointdecoratedinputsegmentlast}{\pgfpoint{1pt}{1pt}}
        }
    }
    \state{final}{}
}

\floatstyle{ruled}
\newfloat{Protocol}{h}{pro}
\newfloat{Procedure}{h}{procedure}
   
\title{\textbf{Coding for interactive communication  correcting \\ insertions and deletions}}
\author{Mark Braverman\thanks{Department of Computer Science, Princeton University, email: mbraverm@cs.princeton.edu. Research supported in part by an   an NSF CAREER award (CCF-1149888), NSF CCF-1215990, a 
Turing Centenary Fellowship, a Packard Fellowship in Science and Engineering, and the Simons Collaboration on Algorithms and Geometry.}
\and 
Ran Gelles\thanks{Department of Computer Science, Princeton University, email: rgelles@cs.princeton.edu} 
\and
Jieming Mao\thanks{Department of Computer Science, Princeton University, email: jiemingm@cs.princeton.edu}
\and
Rafail Ostrovsky\thanks{Department of Computer Science and department of Mathematics, UCLA, email: rafail@cs.ucla.edu. Work supported in part by NSF grants 09165174, 1065276, 1118126 and 1136174, DARPA, US-Israel BSF grant 2008411, OKAWA Foundation Research Award, IBM Faculty Research Award, Xerox Faculty Research Award, B. John Garrick Foundation Award, Teradata Research Award, and Lockheed-Martin Corporation Research Award.  The views expressed are those of the author and do not reflect the official policy or position of the Department of Defense or the U.S. Government.}}

\begin{document}

\maketitle


\begin{abstract}
We consider the question of interactive communication, in which two remote parties perform a computation while their communication channel is (adversarially) noisy. 
We extend here the discussion into a more general and stronger class of noise,
namely, we allow the channel to perform \emph{insertions} and \emph{deletions} of symbols.
These types of errors may bring the parties ``out of sync'', 
so that there is no consensus regarding the current round of the protocol.

In this more general noise model, we obtain the first interactive coding scheme 
that has a constant rate and tolerates noise rates of up to~$1/18-\eps$.
To this end we develop a novel primitive we name 
\emph{edit distance tree code}. The edit distance tree code is designed to replace the Hamming distance constraints in Schulman's tree codes (STOC~93), with a stronger edit distance requirement. However, the straightforward generalization of tree codes to edit distance does not seem to yield a primitive that suffices for communication in the presence of synchronization problems. Giving the ``right'' definition of edit distance tree codes is a main conceptual contribution of this work. 
\end{abstract}



\section{Introduction}

In the setting of interactive communication
two remote parties, Alice and Bob, wish to run a distributed protocol utilizing a noisy communication channel. The study of this problem was initiated by the seminal work of Schulman~\cite{schulman92,schulman93,schulman96} who showed a coding scheme for interactive protocols in which 
the communication complexity of the resilient protocol 
is larger than the communication of the input (noiseless) protocol
by only a constant factor.
Schulman's coding schemes tolerates \emph{random noise} where each bit is flipped with a small probability, as well as some \emph{adversarial noise} where the only restriction on the noise is the amount of bits being flipped by the adversary.
Subsequently, many works considered the question of interactive communication, obtaining   coding schemes  reaching optimality in terms of their
computational efficiency~\cite{GMS11,GMS14,BK12,BN13,BKN14,GH14},
communication efficiency~\cite{KR13,haeupler14, GH15}, 
and 
noise resilience,
both in the standard setting~\cite{BR11,BR14,BE14,GH14}, and in various other noise models and settings~\cite{ORS05,ORS09,FGOS13,FGOS15,GHS14,GH14,AGS13,BNTTU14,EGH15}.

The recent successes in developing the theory of interactive error-correcting codes brought the study of two-way interactive coding to nearly match what we know about good codes against adversarial noise in the one-way setting. 

So far, all works focused on either substitutions (where Eve can substitute a sent symbol with a different symbol from the alphabet) or erasures (where Eve can substitute a sent symbol with a $\bot$). 
In this work we extend the question of coding for interactive communication over noisy channels to a more general type of noise. Namely, we consider channels with \emph{insertions and deletions} (indels). In the one-way setting, this corresponds to the insertion and deletion model, where Eve is allowed to completely remove transmitted symbols, or inject new symbols. 
Note that this model is stronger than the substitution model, since 
a substitution can always be implemented as a deletion followed by an insertion. Even in the one-way regime, this model is more difficult to analyze than the model with substitution errors. 
As an example, Schulman and Zuckerman~\cite{SZ99} gave a polynomial-time encodable and decodable codes for insertion/deletion errors. Their code can tolerate around $\frac{1}{100}$ fraction of insertion/deletion errors. This should be contrasted with efficient codes in the standard noise setting, e.g.,~\cite{justesen72}, tolerating about~$\frac14$ fraction of bit flips.

The major additional challenge in dealing with indels in the interactive setting compared to the non-interactive indel model and the interactive substitutions model, is that we can no longer assume that Alice and Bob are synchronized: at a given time they may be at different stages of their sides of the protocol! Indeed, if Eve deletes Alice's transmission to Bob and additionally injects a `spoofed' reply from Bob back to Alice, then while Bob has received no message and assumes the protocol hasn't advanced yet, 
Alice has received a spoofed reply, and proceeds to the next step of her protocol. From this point and on, unless the insertion/deletion is detected, the parties are unsynchronized, as they run different steps of the protocol. The challenge in dealing with this model is to design a protocol that manages to succeed even without knowing whether the two parties are synchronized.

\subsection{Modeling insertions and deletions}
Some care is required when dealing with insertion and deletion noise patterns, as certain choices make the model too strong or too weak. For example, consider the case where Alice and Bob send each other symbols in an alternating way.
Then, even if a single deletion of a symbol is allowed the noise can cause the protocol to ``hang'': Bob will be waiting for a symbol from Alice, while Alice will be waiting for Bob's response. Clearly, such a model is too strong for our purpose, and we should restrict the allowed noise patterns to preserve the protocol's liveliness.

There are two main paradigms for distributed protocol in which parties are not fully-synchronized. The first is a \emph{message-driven} paradigm, in which each party ``sleeps'' until the arrival of a new message that triggers it into performing some computation and sending a message to the other party. The second is \emph{clock-driven}, where each party holds a clock: each clock tick, the party wakes up, checks the incoming messages queue, performs some computation, and sends a message to the other side. The issue here is that different parties may have mismatching or skewed clocks. Then, instead of acting in an alternating manner, one party may wake up several times while the other party is still asleep. 

We emphasize that if the parties have matching clocks, then no insertions and deletions are possible---channel corruption in this case has either the effect of changing one symbol to another (as in a standard noisy channel), or causing a detectible corruption, i.e., an erasure. Both these types of noise are substantially weaker than insertions and deletions, and were already analyzed in previous work (e.g.,~\cite{schulman96,BR14,FGOS15,EGH15}).

Our noise model, which we describe shortly, makes sense for both the above paradigms: it guarantees liveliness in a message-driven setting; 
for the clock-driven model, we can show that any such settings reduces to our model, that is, any resilient protocol in our model can be used to obtain a resilient protocol in the clock-driven setting. The skewness of the clocks in that case, is related to the noise-resilience of the protocol in our model. See the full version for the complete details.

\smallskip
In this work we assume a message-driven setting, where the parties normally speak in alternating manner. Any corruption in our model must be a deletion which is followed by an insertion. 
We name each such tampering as an \emph{edit corruption}.
\begin{definition}[Edit Corruption]
An edit-corruption is a single deletion followed by a single insertion (whether the inserted symbol is aimed at the same or the opposite party as the deleted message). 
\end{definition}

This gives rise to two types of attacks Eve can perform: 
(i) delete a symbol and replace it with a different symbol (insert a symbol at the same direction as the deleted symbol; a substitution attack);  
(ii) delete a symbol, and insert a spoofed reply from the other side (insert a symbol at the opposite direction of the deleted symbol). 
The second type of corruption has an effect of making the parties `out-of-sync': one party advances one step in the protocol, while the other does not; see Figure~\ref{fig:attacks}. 
Note that a substitution has a cost a single corruption, i.e., it is counted as a single deletion followed by a single insertion. 
Also note that 
although an outside viewer can split Eve's attack into pairs of deletion-and-insertion, 
the string that a certain party receives,
from that party's own view, 
suffers an \emph{arbitrary} pattern of insertions/deletions.

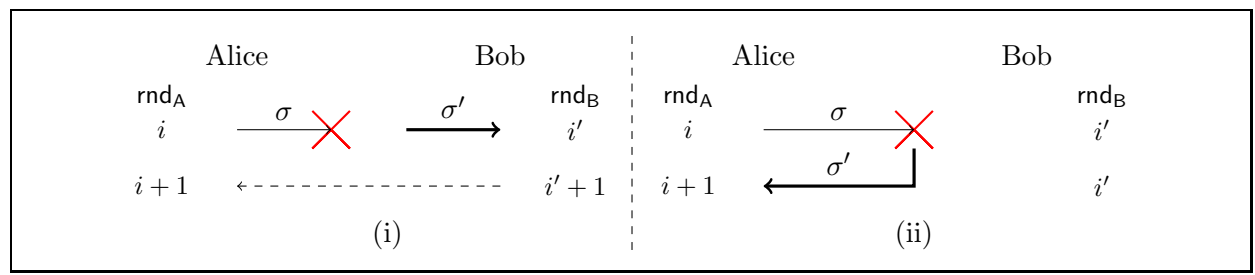
\begin{figure}[ht]
\vspace{1ex plus 1ex}
\begin{center}
\begin{minipage}{1.0\columnwidth}
\begin{framed}
\begin{center}
\begin{tikzpicture}
\draw[->,very thick] (2.25,2.75) -- node[above]{$\sigma'$} (3.5,2.75) ;
\draw[thick,red] (1.25,2.75) -- +(-0.25,-0.25) -- +(.25,.25);
\draw[thick,red] (1.25,2.75) -- +(0.25,-0.25) -- +(-.25,.25);
\draw[->]  (0,2.75)  -- node[above] {$\sigma$} (1.25,2.75);

\draw[thick,red] (9,2.75) -- +(-0.25,-0.25) -- +(.25,.25);
\draw[thick,red] (9,2.75) -- +(0.25,-0.25) -- +(-.25,.25);
\draw[->]  (7,2.75)  -- node[above] {$\sigma$} (9,2.75);
\draw[->,very thick]  (9,2.5) -- (9,2) -- node[above] {$\sigma'$} (7,2);

\draw[dashed] (5.25,1.15) -- (5.25,4);
\node at (0,3.75) {Alice};
\node at (3.5,3.75) {Bob};
\node at (7,3.75) {Alice};
\node at (10.5,3.75) {Bob};

\node at (2,1.35) {(i)};
\node at (9,1.35) {(ii)};

\node at (-1,3.2)  {\small $\mathsf{rnd_A}$};
\node at (-1,2.75)   {\small $i$};
\node at (-1,2)   {\small $i+1$};
\draw[->,dashed] (3.5,2) -- (0,2);  
\node at (4.5,3.2)  {\small $\mathsf{rnd_B}$};
\node at (4.5,2.75)   {\small $i'$};
\node at (4.5,2)   {\small $i'+1$};

\coordinate (x) at (7,0);  
\node at ($(x)+(-1,3.2)$)  {\small $\mathsf{rnd_A}$};
\node at ($(x)+(-1,2.75)$)   {\small $i$};
\node at ($(x)+(-1,2)$)   {\small $i+1$};
\node at ($(x)+(4.5,3.2)$)  {\small $\mathsf{rnd_B}$};
\node at ($(x)+(4.5,2.75)$)   {\small $i'$};
\node at ($(x)+(4.5,2)$)   {\small $i'$};

\end{tikzpicture}
\end{center}\vspace{-1.5em}
\end{framed}
\end{minipage}
\end{center}\vspace{-.5em}
\caption{\small An illustration of the two insertion/deletion attacks: (i) a deletion followed by an insertion in the same direction (a substitution); (ii)  a deletion followed by an insertion to the opposite direction (an out-of-sync attack).
The deleted transmission is marked with a cross, and the inserted transmission is marked with a bold arrow. The dashed arrow denotes a possible (non-interrupted) reply.}
\label{fig:attacks}
\end{figure}


\subsection{Our Results and Techniques}
\textbf{Tree codes with edit distance.} When only a single message is to be transmitted (i.e., in the one-way setting), codes that withstand insertions and deletions were first considered by Levenshtein~\cite{levenshtein66}. In such codes, each two codewords must be far away in their \emph{edit distance}, a notion of distance that captures the amount of insertions/deletions it takes to convert one codeword to another. The edit distance replaces the Hamming distance, which essentially counts the amount of bit flips required to turn one codeword to another.

A key ingredient in interactive-communication schemes is the \emph{tree code}~\cite{schulman96},
a labeled tree such that the labels on each descending path in the tree can be seen as a codeword. 
The tree code is parametrized by a distance parameter $\alpha$, and it holds that any two codewords whose paths diverge from the same node, are at least $\alpha$-apart in their Hamming distance. 
Encoding a message via a tree code allows a party to \emph{eventually} obtain the message sent by the other side, as long as not too many errors have happened.
This in turn allows the parties to correct errors that previously occurred in the simulation, and  revert the simulation back into a correct state~\cite{schulman96,BR14}. 
In order to keep the communication overhead a constant factor when tree code encoding is used in interactive schemes, it is required that each label of the tree comes from an alphabet of constant size (that is, the size of the alphabet is independent of the tree's depth and thus independent of the length of protocol to simulate).

It is only natural to believe that we could obtain interactive-communication schemes that withstand insertions/deletions by replacing tree codes with a stronger notion of codes, namely, \emph{edit distance tree codes}. 
In edit distance tree codes, each two codewords (possibly of different lengths) which diverge at a certain point,  are required to be far apart in their edit distance rather than their Hamming distance. 
Yet, since the parties are not synchronized, new difficulties arise.
To give a simple example, assume Alice sends one of the two
following encodings
$s_1=ABC AAA BBB$ and $s_2=ABC ABC AAA$, and assume Bob has
receives the string $ABC BBB$.
If Bob knew that Alice thinks she is in round $6$ of the protocol,
he would decode to~$s_2$;
if he knew that Alice thinks she is in round~$9$, he would decode
to~$s_1$. Alas, he does not know which is the case!


To mitigate situations in which not being synchronized may hurt us,  we require an even stronger property, namely, we wish that the suffixes (of arbitrary lengths) of any two overlapping codewords will have appropriately large edit distance (See Definitions~\ref{def:lambda} and~\ref{def:edtc} for the precise condition). 
This stronger property 
guarantees that two ``branches'' in the tree are far apart in their edit distance, even when they are shifted with respect to each other due to lack of synchronization possibly caused by previous indels. 
We can then show that as long as not too many indels occurred in the \emph{suffix} of the received codeword, the tree-code succeeds to recover the entire sent message.
Crucial in this approach is a notion of distance 
we call \emph{suffix distance} (Definition~\ref{def:SD}),
that measures the amount of noise in a codewords' suffix.
This generalizes a distance measure by Franklin et al.~\cite{FGOS13,FGOS15} (see also Braverman and Efremenko~\cite{BE14}) to the case where the received word may be misaligned with respect to the sent word, due to indels.


%
Alas, while (Hamming distance) tree codes over a constant alphabet were shown to exist by Schulman~\cite{schulman96}, it is not clear if such trees exist for edit distance, and if so, for which distance parameter~$\alpha$, as Schulman's proof doesn't carry over to the edit distance case. Our first result (Section~\ref{sec:edtc}) shows the existence of edit distance tree codes, for any distance parameter~$\alpha$,
\begin{theorem}[Informal]\label{thm:edtc-exist}
For any $\alpha <1$ and any $d,n\in \mathbb{N}$ 
there exists a $d$-ary edit distance tree code of depth~$n$ over a constant-size alphabet.
\end{theorem}
In section~\ref{sec:edtc} we prove the existence of such trees and prove the above lemma. Similarly to \cite{schulman96}, the proof is inductive in the depth of the tree, but the inductive statement needs to be further strengthened in order for the induction to go through.

As in the case of standard tree codes, finding an \emph{efficient} construction for such trees remains an important open question.
Building on the techniques of Gelles, Moitra and Sahai~\cite{GMS11,GMS14} we give 
in Appendix~\ref{sec:potent} an efficient randomized construction of a relaxed notion for edit-distance tree codes, we call a \emph{potent edit distance} tree codes  (see Definition~\ref{def:PEDTC} for the precise definition of potent edit-distance tree codes). 
These trees satisfy the edit-distance guarantee \emph{almost} everywhere, 
and are good enough to replace the tree-code notion of Theorem~\ref{thm:edtc-exist} in most applications.
\begin{theorem}[Informal]\label{thm:pedtc-exist}
For any $\alpha <1$ and any $d,n \in \mathbb{N}$ 
there exists a construction of a $d$-ary potent edit distance tree code of length~$n$ over a constant-size alphabet. The construction is efficient and succeeds with overwhelming probability (in~$n$).
\end{theorem}
While in the rest of the paper (namely, for our coding scheme) 
we assume the edit-distance notion of Theorem~\ref{thm:edtc-exist}, all our schemes work the same when the tree is replaced with a potent one, see Appendix~\ref{sec:PEDTCpro} for further details.
\\
\\
\textbf{Interactive-communication schemes tolerating insertions/deletions.}
Equipped with edit distance tree codes, 
we show a protocol that  
solves the \emph{pointer jumping problem} over a noisy channel with insertions and deletions and exhibits linear communication complexity in the noiseless communication.
Since the pointer jumping problem is complete for two-party interactive communication, this implies a  coding scheme that can simulate any protocol over a channel that may introduce insertion/deletions. 
Specifically, in Sections~\ref{sec:poly} and~\ref{sec:constant}
we show that for any $\eps>0$ and any noiseless protocol~$\pi$ and inputs $x,y$, 
there is a scheme that correctly simulates $\pi$ (that is, produces the transcript $\pi(x,y)$ at both parties), withstands $1/18-\eps$ fraction of edit-corruptions, and has a linear communication complexity with respect to the communication of the protocol~$\pi$.

\begin{theorem}\label{thm:main}
For any $\eps>0$, and for any binary (noiseless) protocol~$\pi$ with communication $CC(\pi)$, there exists a noise-resilient coding scheme with communication $O_\eps(CC(\pi))$ that succeeds in simulating~$\pi$ as long as the adversarial 
edit-corruption rate is at most~$1/18-\eps$.
\end{theorem}

Our coding scheme and analysis follows ideas by Braverman and Rao~\cite{BR14} and by Braverman and Efremenko~\cite{BE14}---first focusing on channels with  polynomial-size alphabet and then generalizing to channels with constant-size alphabet---however, the analysis in the light of insertions and deletions is more complicated and subtle. In particular, similar to~\cite{BE14}, our analysis uses the notion of suffix distance  
for relating the effect of the noise to the progress of the simulation.

We note again that due to out-of-sync attacks, 
it is possible that the parties' belief of the ``current'' round of the protocol is different.
In the worst case, while Alice reaches the end of the coding protocol (say, round~$N$), it is possible that Bob has only reached round~$(1-2\rho)N$, e.g., due to $2\rho N$ deletions in his received communication ($\rho$ is the fraction of edit-corruptions in that instance).
Therefore, if we wish to tolerate a $\rho$-fraction of edit-corruptions, 
it is imperative that the parties output the correct answer already at round $(1-2\rho)N$. 
Our coding scheme (Theorem~\ref{thm:main}) satisfies even this more strict requirement.

Finally, in Appendix~\ref{sec:upperbound} we show that for a family of rigid protocols, in which we require \emph{both} parties to output the correct value at round $(1-2\rho)N$, 
then $\rho=1/6$ is an upper bound on the admissible edit-corruption rate.
\begin{theorem}[Informal]\label{thm:main-upper}
If both parties are required to give output at round $(1-2\rho)N$, then no coding scheme of length $N$ can tolerate an edit-corruption rate of~$\rho=1/6$. 
\end{theorem} 

Closing the gap between the upper bound of $1/6$, and the resilience~$1/18$ achieved by the scheme of Theorem~\ref{thm:main} is left for future work.

\section{Preliminaries}
\label{sec:prelim}

For any finite set $S$, we denote by $x\uniform S$ the case that $x$ is uniformly distributed over~$S$. All logarithms are taken to base~$2$. We denote the set $\{1,2,\ldots, n\}$ by~$[n]$. For a set $\Sigma$ we denote $\Sigma^{\le n} = \cup_{0\le i\le n}\Sigma^i$, and $\Sigma^*= \cup_{i\ge 0}\Sigma^i$.
Let $s \in \Sigma^l$ be a string of length $|s|=l$.
For $1 \leq i \leq j \leq l$, we use $s[i]$ to denote the $i$-th symbol of $s$ and $s[i..j]$ to denote the string $s[i] \circ s[i+1] \circ \cdots \circ s[j]$. 

\begin{definition}[Pointer Jumping Problem]
Any communication protocol of $T$ rounds where the parties alternately exchange bits can
be reduced to the following \emph{pointer jumping problem}~$PJP(T)$: 
Let $\T$ be a  binary tree of depth $T$. 
Alice's input $X$ is a set of \emph{consistent} edges leaving vertices at even depths. 
Bob's input~$Y$ is a set of consistent edges leaving odd-level vertices. A set of edges is \emph{consistent} on a specified set of vertices, if it contains exactly one edge leaving every vertex in that specified set.  
Due to being consistent on all the nodes, $X\cup Y$ defines a unique root-to-leaf path. The parties' goal is to output this unique path. Note that $T$ alternating rounds of communication suffice to compute this path, assuming noiseless channels.
\end{definition}

\begin{definition}[Protocols]
An interactive protocol~$\pi$ for a function $f(x,y)$ is a distributed algorithm that dictates for each party, at every round, the next message (symbol) to send given the party's input and the messages received so far. Each transmitted symbol is assumed to be out of a fixed alphabet~$\Sigma$. The protocol runs for $N$ rounds (also called the length of the protocol), after which the parties give output. An instance of the protocol, on inputs $x,y$ is said to be \emph{correct} if both parties output $f(x,y)$ at the end of the protocol.
\end{definition}

In this work we focus on \emph{alternating} protocols, where as long as there is no noise, 
the protocol runs for $2N$ rounds in which Alice and Bob send symbols alternately.
In the presence of $\rho$-fraction of edit corruptions (i.e., at most $2\rho N$ insertion/deletion errors), it is possible that some party receives only $N(1-2\rho)$ symbols throughout the protocol. We say that the protocol is correct in presence of $\rho$-fraction of edit corruption if both parties output $f(x,y)$ 
by round $N(1-2\rho)$ (according to their own round counting, which may differ from the count of the other party).  

\begin{definition}[String Matching]
We say that $\tau = (\tau_1, \tau_2)$ is a string matching between a sent message~$sm$ and a received message~$rm$ (denoted $\tau:sm \rightarrow rm$), if $|\tau_1| = |\tau_2|$, $del(\tau_1) = sm$, $del(\tau_2) = rm$, and $\tau_1[i] \approx \tau_2[i]$ for all $i = 1,\ldots,|\tau_1|$. Here $del$ is a function that deletes all the $*$'s in the string and two characters $a$ and $b$ satisfy $a \approx b$ if $a=b$ or one of $a$ and $b$ is $*$. We assume that $*$ is a special symbol that does not appear in $sm$ and $rm$.
\end{definition}

\begin{definition}[Edit Distance]
The edit distance of between $sm\in \Sigma^*$ and $rm \in \Sigma^*$ is defined as $
ED(sm,rm) = \min_{\tau:sm \rightarrow rm} sc(\tau_1) + sc(\tau_2).
$
Here $sc(\tau_1)$ is the number of $*$'s in $\tau_1$.
\end{definition}


\begin{fact}[Triangle Inequality]
\label{fact:ed+}
For any three strings $x,y,z$, 
$
ED(x,y) \leq ED(x,z) + ED(y,z).
$
\end{fact}

\begin{fact}
\label{fact:lcs}
For any two strings $x,y$,
$
ED(x,y) = |x| + |y| - 2\cdot LCS(x,y).
$
Here $LCS(x,y)$ is the longest common substring of $x$ and $y$.
\end{fact}

\begin{lemma}
\label{lem:edpr}
Let $x \in \Sigma^m$ be some given string, $m \geq n$ and $|\Sigma| \geq 4$. For any constant $\alpha \in [0,1]$, 
$
Pr_{y \uniform \Sigma^n} [ED(x,y) \leq \alpha \cdot m] \leq |\Sigma|^{-\frac{(1-\alpha) m}{2}}.
$
\end{lemma}
\begin{proof}
By Fact \ref{fact:lcs}, we know that $ED(x,y) \leq \alpha \cdot m \Leftrightarrow LCS(x,y) \geq \frac{n+(1-\alpha)m}{2}$. If we want to upper bound the number of $y \in \Sigma^n$ such that $LCS(x,y)\geq \frac{n+(1-\alpha)m}{2}$, we can just enumerate locations of the common substring. Then,
\(
\Pr_{y \uniform \Sigma^n} [ED(x,y) \leq \alpha \cdot m]  
= \Pr_{y \uniform \Sigma^n} \left[LCS(x,y) \geq \frac{n+(1-\alpha)m}{2}\right] 
\leq \sum_{l =\frac{n+(1-\alpha)m}{2}}^n \left(\binom{n}{l} |\Sigma|^{n - l}\right)/|\Sigma|^n
\leq 2^{n} \cdot |\Sigma|^{- \frac{n+(1-\alpha)m}{2}} 
\leq |\Sigma|^{- \frac{(1-\alpha)m}{2}}.
\)
\end{proof}


\section{Edit-distance tree code}
\label{sec:edtc}

In this section we recall the notion of tree-codes~\cite{schulman96} and provide a novel primitive, namely, the edit-distance tree-code. 

\begin{definition}[Prefix Code]
A prefix code $C:\Sigma_{in}^n \rightarrow \Sigma_{out}^n$ is a code such that $C(x)[i]$ only depends on $x[1..i]$.  $C$ can be also considered as a $|\Sigma_{in}|$-ary tree of depth $n$ with symbols written on edges of the tree using an alphabet size $|\Sigma_{out}|$. On this tree, each tree path from the root of length $l$ corresponds to a string $x \in \Sigma_{in}^l$ and the symbol written on the deepest edge of this path corresponds to $C(x)[l]$. 
\end{definition}

\begin{definition}[$\alpha$-bad Lambda]
\label{def:lambda}
We say that a prefix code $C$ contains an $\alpha$-bad lambda if when you consider this prefix code as a tree, there exist four tree nodes $A$,$B$,$D$,$E$ such that $B \neq D$, $B \neq E$, $B$ is $D$ and $E$'s common ancestor, $A$ is $B$'s ancestor or $B$ itself, and  $ED(AD,BE) \leq \alpha \cdot \max(|AD|, |BE|)$. Here $AD$ and $BE$ are strings of symbols along the tree path from $A$ to $D$ and the tree path from $B$ to $E$. See Figure~\ref{fig:lambda} for an illustration.
\end{definition}

\begin{definition}[Edit-distance Tree Code]\label{def:edtc}
We say that a prefix code $C$: $\Sigma_{in}^n \rightarrow \Sigma_{out}^n$ is a $\alpha$-edit-distance tree code if $C$ does not contain an $\alpha$-bad lambda.
\end{definition}

Our main theorem in this section is the existence of edit-distance tree codes, 
\begin{theorem}
\label{thm:edtc}
For any $d \geq 2$, $n >0$ and $0<\alpha <1$, there exists an
$\alpha$-edit-distance tree code  of depth~$n$ 
with alphabet size
$|\Sigma_{in}|=d$ and 
$|\Sigma_{out}|=(176 \cdot d)^{4/(1-\alpha)}$.
\end{theorem}

\begin{proof}
We prove this theorem by induction on~$n$. To this end we define 
a slightly stronger notion than $\alpha$-bad lambda free, which we call ``excellent". Intuitively, if a tree-code $C$ is excellent, then with a good probability, $C$ will not cause a bad lambda in a tree-code that contains $C$ as a subtree. This would allow us to  construct lambda-free trees of length $n$ building on lambda-free trees with a smaller depth as subtrees.

\begin{definition}[Potential Probability]
For any prefix code $C:\Sigma_{in}^n \rightarrow \Sigma_{out}^n$, and any $i\ge0$,
consider $C$ as a tree and define a new (non-regular) tree $C'$ by connecting a simple path of length $i$ to the root of $C$, making the other end of this path the root of $C'$. Label each edge along the new path with a symbol from~$\Sigma_{out}$ chosen  uniformly and independently.

The potential probability $P_i(C)$ is defined as the probability that the new tree $C'$ has a $\alpha$-bad lambda with $A=$ the root of $C'$ and $B=$ the root of $C$. See Figure~\ref{fig:potential} for an illustration.
\end{definition}

\begin{definition}[Excellent]
\label{def:exc}
For a constant $c_1>1$, which we will fix shortly, we say that a prefix code~$C$ is excellent if $\forall i \geq 0$, 
$P_i(C) \cdot (d \cdot c_1)^i < 1$  and $C$ is $\alpha$-bad-lambda-free. 
\end{definition}

In the proof, we are going to construct a set called $S_n$ which would contain only excellent $d$-ary prefix codes of depth~$n$ with alphabet size $d^{O_{\alpha}(1)}$. Since excellent is a stronger notion than $\alpha$-bad lambda-free, if we can construct $S_n$ and show that for any $n>0$ it is non-empty, then we are done. 
In each $S_n$, all codes use the same $\Sigma_{in}$ and $\Sigma_{out}$. We have $|\Sigma_{in}| = d$ and $|\Sigma_{out}| =s$
where the conditions we put on~$s$ thorough the following proof are given by:
\[
s = \max \left\{ \frac{\alpha d}{(1-\alpha)} \cdot (d\cdot c_1)^{\frac{\alpha}{1-\alpha}} + 1, d^2, (\frac{4d}{c_2})^{\frac{4}{1-\alpha}}, (c_1 \cdot d \cdot 4)^{\frac{4}{1-\alpha}} \right\}.
\]
Here $c_1 = 44$, and $c_2 = \frac{1}{11}$.  Therefore, taking $s\ge (176d)^{4/(1-\alpha)}$ satisfies all the above conditions. 

\medskip
Let us now inductively construct~$S_n$. For notation convenience, let $S_0 = \{ \text{a single node}\}$. For $n > 0$, let
\[
\tilde{S}_n =\left\{
T = 
\begin{tikzpicture}[baseline=(current bounding box.center)]
\node[circle,fill=black,minimum size=4pt,inner sep=0pt] at (0,0) {};
\node[isosceles triangle,shape border rotate=90,draw=black] (first) at (-2,-1.69)  {$T_1$};
\node[isosceles triangle,shape border rotate=90,draw=black] (second) at (-0.5,-1.69)  {$T_2$};
\node[isosceles triangle,shape border rotate=90,draw=black] (third) at (2,-1.69)  {$T_d$};
\node at (0.9, -1.2) {$\cdots$};
\draw [shorten >=0pt] (0,0)  --  node[above left] {$\sigma_1$ } ($(first.north)-(0,0.015)$);
\draw (0,0) -- node[below right] {$\sigma_2$} ($(second.north)-(0,0.015)$) ;
\draw (0,0) -- node[above right] {$\sigma_d$} ($(third.north)-(0,0.015)$) ;
\end{tikzpicture}
   \ \middle|\ T_1,T_2,...,T_d \in S_{n-1} \text{ and } \sigma_1,\sigma_2,...,\sigma_d\in \Sigma_{out}\right\}.
\]
and define
$
S_n = \{C \in \tilde{S}_n | ~C \text{ is excellet}\}.
$
From this definition, we directly have that every $C \in S_n$ is excellent, and we are only left to show that $S_n$ is non-empty. Actually, we are going to prove the following claim by induction.\\ 

\begin{claim}
For all $n$, $\frac{|S_n|}{|\tilde{S}_n|} \geq c_2 = \frac{1}{11}$.
\end{claim}

\noindent\textbf{Base case ($n=1$):} Consider the following set.
\[
S'_1 = \{C \mid C \text{ is a $d$-ary prefix code of depth 1 
and $d$ different codewords}\}
\]
We want to show that $S'_1 \subseteq S_1$. For any $C\in S'_1$, it is clear that $C$ does not have any $\alpha$-bad lambda. Because in this depth 1 tree, one can only pick $A=B$ to be the root, and $D,E$ to be some different leaves. Then $ED(AD,BE) = 2$ and $|AD| = |BE| =1$. So $ED(AD,BE) > \alpha \cdot \max(|AD|,|BE|)$. Now let's consider the potential probability $P_i(C)$. Suppose there exists an $\alpha$-bad lambda in the tree after adding a path of length $i$. Then in this $\alpha$-bad lambda, $B$ is the original root,  $AB$ is the path added to the root and $D$ and $E$ are two different leaves of the tree. There are two cases to consider:
\begin{enumerate}
\item $i = |AB| > \alpha / (1-\alpha)$: In this case, since $|BE| = 1$, $(|AD| - |BE| )/ |AD| = 1 - \frac{1}{1+|AB|} > \alpha$. So it is not possible that $ED(AD,BE) < \alpha \cdot |AD|$. Thus  $P_i(C) = 0$. 
\item $i = |AB| \leq \alpha / (1-\alpha)$: In this case, let $W$ be the event that 
one of the labels along the path $AB$ equals to one of the $d$ codewords of $C$. 
Clearly, when $W$ does not happen, $ED(AD,BE) = |AD|+|BE| > \alpha \cdot |AD|$. 
It is also immediate that $\Pr[W] \leq i\cdot\frac{d}{s}$. We require $s >  \frac{\alpha d}{(1-\alpha)} \cdot (d\cdot c_1)^{\frac{\alpha}{1-\alpha}} $, and get that
$
P_i(C) \cdot (d\cdot c_1)^i \leq \Pr[W]\cdot (d\cdot c_1)^i \leq \frac{\alpha d}{(1-\alpha)s} \cdot (d\cdot c_1)^{\frac{\alpha}{1-\alpha}} < 1.
$
\end{enumerate}
Therefore, we have proved that all the prefix codes in $S'_1$ are excellent and therefore in $S_1$. Then because $s \geq d^2$, we get
$
\frac{|S_1|}{|\tilde{S}_1|} \geq \frac{|S'_1|}{|\tilde{S}_1|} = \frac{s(s-1) \times \cdots \times (s-d+1)}{s^d} \geq (1-1/d)^d > 1/11 = c_2.
$\\

\noindent\textbf{Inductive step:} Suppose we already know that for $1 \leq i < n$, $\frac{|S_i|}{|\tilde{S}_i|} \geq c_2$, and let us prove that $\frac{|S_n|}{|\tilde{S}_n|} \geq c_2$.
We will use the following lemma, whose proof 
is quite straightforward and can be found in Appendix~\ref{app:treecode}.\looseness=-1
\begin{lemma}
\label{lem:disuni}
If for all $1 \leq i < n$, $\frac{|S_i|}{|\tilde{S}_i|} \geq c_2$. Then for any $x \in \Sigma_{out}^j$ where $ 0< j \leq n$ and $y \in \Sigma_{in}^n$, it holds that
$
\Pr_{C \uniform \tilde{S}_n} [ C(y)[1..j] = x] \leq c_2^{1-j} s^{-j}.
$
\end{lemma}

In order to show that $\frac{|S_n|}{|\tilde{S}_n|} \geq c_2$, we can choose $C$ randomly from $\tilde{S}_n$, and show that ${\Pr[C \text{ is excellent}] \geq c_2}$. 
To this end, consider the conditions of being excellent in turn:
\begin{enumerate}[itemsep=3pt]
\item Let's first consider $\Pr[P_i(C) \cdot (d\cdot c_1)^i \geq 1]$. By Lemma \ref{lem:disuni} and Lemma \ref{lem:edpr}, we get
\begin{eqnarray*}
&&\E_{C \uniform \tilde{S}_n}[P_i(C)] \\
&\leq& \sum_{1 \leq n_1, n_2 \leq n} \ \sum_{\substack{x \in \Sigma_{in}^{n_1}, y \in \Sigma_{in}^{n_2}, \\ x[1] \neq y[1]}} \Pr_{\substack{z \uniform \Sigma_{out}^i \\ C\uniform \tilde S_n}}[ED(z\circ C(x), C(y)) \leq \alpha \cdot \max(i + n_1, n_2)]\\
&\leq& \sum_{1 \leq n_1, n_2 \leq n} d^{n_1 + n_2} \cdot c_2^{2-n_1-n_2}\Pr_{x \uniform \Sigma_{out}^{n_1},\, y \uniform \Sigma_{out}^{n_2},\, z \uniform \Sigma_{out}^i}[ED(z\circ x, y) \leq \alpha \cdot \max(i + n_1, n_2)] \\
&\leq&  \sum_{1 \leq n_1, n_2 \leq n} \left(\frac{d}{c_2}\right)^{n_1 + n_2} s^{-\frac{1-\alpha}{2} \cdot \frac{n_1 + n_2 + i}{2}} \leq \sum_{1 \leq n_1, n_2 \leq n} \left(\frac{d}{c_2}\cdot s^{-\frac{1-\alpha}{4}}\right)^{n_1 + n_2}  s^{-\frac{(1-\alpha)i}{4}}\\
&\leq&  \sum_{1 \leq n_1, n_2 \leq n}  \left(\frac{1}{4}\right)^{n_1+n_2} (c_1 \cdot d \cdot 4)^{-i} \leq \sum_{j =2}^{\infty}\frac{j}{4^j} \cdot (c_1 \cdot d \cdot 4)^{-i} = \frac{7}{36} (c_1 \cdot d \cdot 4)^{-i} 
\end{eqnarray*}
By Markov's inequality, 
$
\Pr[{P_i(C) \cdot (d\cdot c_1)^i \geq 1}] \leq \frac{7}{36} \cdot \frac{1}{4^i}.
$
 Then 
$
\sum_{i=0}^{\infty} \Pr[P_i(C) \cdot (d\cdot c_1)^i \geq 1] \leq \frac{7}{36} \sum_{i=0}^{\infty} \frac{1}{4^i} = \frac{7}{27} < \frac{5}{11} = \frac{1 - c_2}{2}.
$
\item Now let's consider the event that $C$ has an $\alpha$-bad lambda with $A=B=root$. It is easy to see that this event is equivalent to $P_0(C) \geq 1$. Therefore it is covered by the previous case.
\item Now let's consider the event that $C$ has an $\alpha$-bad lambda with $A=root$ and $B \neq root$. By Lemma \ref{lem:disuni} and Definition \ref{def:exc}, we have
\begin{multline*}
Pr[C \text{ has an $\alpha$-bad lambda with $A=root$, $B \neq root$}]  \\
\leq \sum_{n_1 \leq n} d^{n_1} \cdot c_2^{1-n_1} \cdot (d \cdot c_1)^{-n_1}  \leq \sum_{n_1 = 1}^{\infty} (\frac{1}{c_1c_2})^{n_1} = \sum_{n_1 = 1}^{\infty} (\frac{1}{4})^{n_1}  = \frac{1}{3} < \frac{5}{11}  = \frac{1 - c_2}{2}.
\end{multline*}
\item Finally let's consider the case that $C$ has an $\alpha$-bad lambda with $A \neq root$. It is easy to see that this event never happens because $C$'s subtrees rooted at depth 1 in $C$ are all excellent.
\end{enumerate}
Therefore by union bound,  $\Pr[C \text{ is excellent}]$ is at least
\begin{align*}
\ge & 1 -\sum_{i=0}^{\infty} \Pr[P_i(C) \cdot (d\cdot c_1)^i \geq 1] - \Pr[C \text{ has an $\alpha$-bad lambda with $A=root$, $B \neq root$}]\\
\ge&  1 - \frac{1 - c_2}{2} -\frac{1 - c_2}{2} = c_2. \qedhere
\end{align*}
\end{proof}

\subsection{Decoding of edit-distance tree codes}
The decoding of a codeword $rm\in \Sigma^j_{out}$ via an edit-distance tree-code~$C$ amounts to finding a message 
whose encoding minimizes the suffix distance to~$rm$, i.e.,
$
\textsf{DEC}(rm) = \text{arg\,min}_{m\in\Sigma^{\le n}_{in}} \ SD( rm, C(m)),
$
where the suffix distance $SD(\cdot,\cdot)$ is defined as follows.
\begin{definition}[Suffix Distance]
\label{def:SD}
Given any two strings $sm,rm\in \Sigma^*$,
the suffix distance between $sm$ and $rm$ is  
$
SD(sm,rm) = \min_{\tau:sm \rightarrow rm} \max_{i=1}^{|\tau_1|} \frac{sc(\tau_1[i..|\tau_1|]) + sc(\tau_2[i..|\tau_2|])}{|\tau_1| - i + 1 -  sc(\tau_1[i..|\tau_1|])}.
$
\end{definition}

The following lemma, which plays an important role in the analysis of the simulation protocols that we present in the next sections, 
shows that if a message $sm$ is encoded by some $\alpha$-edit-distance tree code and the received message $rm$ satisfies ${SD(sm,rm) \leq \frac{\alpha}{2}}$, then the receiver can recover the entire sent message correctly. 
\begin{lemma}
\label{lem:unidec}
Let $C$: $\Sigma_{in}^n \rightarrow \Sigma_{out}^n$ be an $\alpha$-edit-distance tree code, and let $rm \in \Sigma_{out}^m$ ($m$ can be different from $n$). There exists at most one $sm \in \cup_{i=1}^n C(\Sigma_{in}^n)[1..i]$ such that $SD(sm,rm) \leq \frac{\alpha}{2}$.
\end{lemma}

The proof appears in Appendix~\ref{app:treecode}.




\section{A coding scheme with a polynomial alphabet size}
\label{sec:poly}
In this section, we show a protocol~$\pi$ 
that 
solves $PJP(T)$ in $O(T)$ rounds over channels with
alphabet size $poly(T)$,  
and is resilient to (a constant fraction of) insertion/deletion errors. 
Since the $PJP(T)$ is complete for interactive communication, this implies that any  binary protocol with $T$ rounds can be simulated in $O(T)$ rounds over a channel with polynomially-large alphabet that corrupts at most a fraction $1/18-\eps$ of the transmissions. 
While this protocol does not exhibit a constant rate, it contains all the main ideas for the constant-rate protocol of  Theorem~\ref{thm:main}, and thus we focus on this simple variant first. Then, in Section \ref{sec:constant} we discuss how to reduce the alphabet size and achieve a protocol with $O(T)$ communication complexity with the same resilience guarantees. 

\medskip
Assume $\pi$ has $2N$ alternating rounds, that is, Alice and Bob send $N$ symbols each, assuming there are no errors.
We would like the  protocol to resist a fraction of $\rho$ edit-corruptions, that is, the protocol should succeed as long as there are at most 
$2\rho N$ insertion/deletion errors. Due our assumption that the adversary never causes the protocol to ``get stuck'', this amounts to at most $2\rho N$ deletions, where each deletion is followed by an insertion. 

We assume that Alice and Bob share some fixed $\alpha$-edit-distance tree code $C:\Sigma_{in}^{N} \rightarrow \Sigma_{out}^{N}$ given by Theorem \ref{thm:edtc}. We will set the values of $N$, $\alpha$, and $\Sigma_{in}$ later. Currently we only need to know $N = poly(T)$. 

Let us begin with a high-level outline of the protocol~$\pi$.
The protocol  basically progresses by sending edges in the tree $\T$ of the underlying~$PJP(T)$, interactively constructing the joint path (similar to~\cite{BR14}). 
%
To communicate an edge~$e$, the parties encode it as a pair of numbers $(n,s)$, where
$0 \leq n \leq N$ and $s\in\{ 1,2,3,4\}$.
The value $n$ indicates the number of some previous round in which some edge~$e'$  was sent, and the value $s$ determines $e$ as the $s$-grandchild of $e'$.
That is, we always send an edge $e$ by linking it to an edge $e'$ that was previously sent, such that $e'$ is located two levels above in unique path leading from the root to $e$. 
If $e$ does not have a grandparent (e.g., it is the at the first or second level in~$\T$), we will set $n = 0$. 
Sometimes the parties have no edge to send, 
in which case they set $n = N$ and say that $e$ is an empty edge.
We take $\Sigma_{in}$ to be all the possible encodings $(n,s)$. As $ 0 \leq n \leq N$ and $1 \leq s \leq 4$, we have $|\Sigma_{in}| = poly(N) = poly(T)$.  As $|\Sigma_{out}|$ is polynomial in $|\Sigma_{in}|$ (Theorem~\ref{thm:edtc}), we also have $|\Sigma_{out}| = poly(N) =poly(T)$.



The protocol~$\pi$ is described in Protocol~\ref{pro:pi}.
The description is for Alice side; Bob's part is symmetric. Here we explain more than the pseudocode on how to get $E(d_A)$. Basically $d_A$ is a string of symbols in $\Sigma_{in}$ and those symbols are in the form $(n,s)$. $E(d_A)$ will be the set of edges these symbols in $d_A$ represent. To get the edge each symbol $(n,s)$ represents, we will first find the edge sent in $n$-th round and get its proper grandchild according to $s$. If $n$ is not in the correct range then we consider $d_A$ as not valid. 
\begin{Protocol}
\caption{The protocol~$\pi$} 
\label{pro:pi}
\small
Let $\T$ be given by $PJP(T)$. Recall that Alice's input is $X$. Assume the parties share some fixed $\alpha$-edit-distance tree code $C:\Sigma_{in}^{N} \rightarrow \Sigma_{out}^{N}$. 

\medskip
Initially we set the counter $i = 0$. 
For any leaf node~$v$ in~$\T$ we initialize a counter $s(v)=0$.
Run the following for $N$ times. 
\begin{enumerate}
\item $ i \leftarrow i  + 1$. 
\item Receive a symbol $r_A[i]$ from the other party. (For Alice, if $i = 1$, skip this step)
\item Find $d_A \in \Sigma_{in}^*$ which minimizes $SD(d_A,r_A[1...i])$. 
\item If $E(d_A) \cup X$ has a unique path from the root in~$\T$, do the following.  Here $E(d_A)$ is the set of edges indicated by $d_A$, if $d_A$ is not a valid string of symbols, $E(d_A) = \emptyset$. 
\begin{enumerate}
\item If this path reaches a leaf node~$v$, then $s(v) \leftarrow s(v) + 1$. 
\item Let $e$ be the deepest edge on the the unique path from root. If $e \in X$ , and $e$ is either an edge in the first or second level of $\T$ or $e$'s grandparent has been sent, set $s_A[i]$ to be encoding of $e$, otherwise set $s_A[i]$ to be encoding of an empty edge. 
\end{enumerate} \vspace{-0.5em}
\item If $E(d_A) \cup X$ does not have a unique path from the root in~$\T$, set $s_A[i]$ to be encoding of an empty edge.
\item If $i = N (1- 2\rho)$, output the leaf node $v$ with the largest $s(v)$. 
\item Send $C(s_A[1...i])[i]$ to Bob. 
\end{enumerate}
\end{Protocol}

We now analyze Protocol~\ref{pro:pi} and prove it resists up to $(1/18-\varepsilon)$-fraction of edit corruptions.
Let $N_A$ and $N_B$ be the counter~$i$ of Alice and Bob respectively, when one of them reaches the end of the protocol~$\pi$. 
Let $\tau_A = (\tau_1, \tau_2)$ be the string matching between $s_B[1..N_B]$ and $r_A[1..N_A]$ that is consistent with the protocol. Let $\tau_B = (\tau_3, \tau_4)$ be the string matching between $s_A[1..N_A]$ and $r_B[1..N_B]$. Recall that we use $sc(\tau)$ to denote the number of $*$'s in the string. By definition, $sc(\tau_1) + sc(\tau_3)  \leq 2\rho N$ and $sc(\tau_2) + sc(\tau_4) \leq 2\rho N$. 

In the analysis we count the number of rounds in which Alice
correctly decodes the entire (current) set of edges sent by Bob. We call each such round a \emph{good decoding}.
\begin{definition}[Good Decoding]
When a party decodes a message, 
we say it is a \emph{good decoding} if 
the decoded messages is exactly the one sent by the other side (i.e., $d_A = s_B[1..i]$ or $d_B = s_A[1..i]$, assuming the other side is at round~$i$), 
and 
the symbol just received is not an adversarial insertion. 
If a decoding is not good, we call it a bad decoding.
\end{definition}
In the following lemma we show relate the number of good decodings to the noise, and show that as long as noise is small enough, there will be many rounds with good decodings. The proof can be found in appendix \ref{app:poly}.
\begin{lemma}
\label{lem:gooddec}
Alice has at least 
$N_A + (1-\frac2{\alpha})sc(\tau_2) - (1+\frac2\alpha)sc(\tau_1)$
good decodings. Bob has at least 
$N_B+  (1-\frac2{\alpha})sc(\tau_4) - (1+\frac2\alpha)sc(\tau_3)$
good decodings. 
\end{lemma}

After we have established that Alice and Bob will have many good decodings, we show that this implies they will have good progress in constructing their joint path in the underlying $PJP(T)$.
%
Consider the $N_A + N_B$ decodings that happen during the protocol, and sort them in a natural order; we say that these decodings occur at ``times'' $t=1,2,\ldots ,N_A + N_B$. Note that the decodings need not be alternating --- Eve's insertions and deletions may cause one party to perform several consecutive decodings while the other party does not receive any symbol, and performs no decoding.
We also assume that for each decoding, 
the sending of the next symbol happens at the same ``time'' as the decoding.
Let $e_A(t)$ be the set of edges Alice has sent at time $t$ and $e_B(t)$ be the set of edges Bob has sent at time~$t$. Let $P$ be the correct path of length $T$ of~$PJP(T)$. Define $l(t)$ to be the length of the longest path from the root using edges in $P \cap (e_A(t) \cup e_B(t))$. Basically $l(t)$ measures how much progress Alice and Bob have made. Let us also define $m(i)$ to be the first time $t$ such that $l(t) \geq i$.  For notation convenience, let $m(0) = 0$. 

The following lemma shows that if the parties do not make progress, many bad decodings (and thus, many errors) must have occurred.  The proof can be found in appendix \ref{app:poly}.
\begin{lemma}
\label{lem:prog}
For $i = 0,...,T-1$, if $m(i + 1) \neq m(i) +1$, then during time $m(i)+1,...,m(i+1) -1$, the following is true.
\begin{enumerate}
\item If $i$ is odd, then there are no good decodings of Bob. The number of good decodings of Alice is at most the number of bad decodings of Bob.
\item If $i$ is even, then there are no good decodings of Alice. The number of good decodings of Bob is at most the number of bad decodings of Alice.
\end{enumerate}
\end{lemma}

Combining the above lemmas, we get the main theorem for protocols with polynomial size alphabet. 
\begin{theorem}
\label{thm:main1}
For any $\varepsilon>0$,
the protocl~$\pi$ of Protocol~\ref{pro:pi}
with 
$N=\lceil\frac{T}{16\varepsilon}\rceil$, 
and a $(1-\eps)$-edit tree code, 
solves $PJP(T)$ and 
is resilient to a $(1/18-\eps)$-fraction of edit corruptions.
\end{theorem}

\begin{proof}
Set $\rho= \frac{1}{18} - \varepsilon$ and $\alpha=1-\eps$. 
Let $g_A$ be the number of good decodings of Alice, $b_A=N_A-g_A$ be the number of bad decodings of Alice. Similarly, let $g_B$ be the number of good decodings of Bob, and let $b_B=N_B-g_b$.
Recall that $sc(\tau_1)+sc(\tau_3)=sc(\tau_2)+sc(\tau_4)\le 2\rho N$,
then by Lemma~\ref{lem:gooddec}, we have 
\begin{align*}
b_A + b_B 
%
&\le 
	\frac{2}{\alpha}(sc(\tau_1) + sc(\tau_2)) + sc(\tau_1) - sc(\tau_2) +
 	\frac{2}{\alpha}(sc(\tau_3) + sc(\tau_4))  + sc(\tau_3) - sc(\tau_4) \\
&\le 
	\frac{8 \rho N}{\alpha}.
\end{align*}
Then we have, 
\begin{align*}
g_A &= N_A - b_A  \ge N_A -  \frac{8 \rho N}{\alpha} \geq  (N_A -N(1-2\rho)) + N(1-2\rho) -  \frac{8 \rho N}{\alpha} \\ \displaybreak[0]
&\geq  b_A+b_B -  \frac{8 \rho N}{\alpha} - (N_A -N(1-2\rho)) + N(1-2\rho) +   \frac{8 \rho N}{\alpha} \\ \displaybreak[0]
&=  b_A + b_B + (N_A -N(1-2\rho))   + N(1 - {16 \rho }/{\alpha} - 2 \rho) \\ \displaybreak[0]
&> b_A + b_B +  (N_A -N(1-2\rho))  + N(1 - (1-18\varepsilon)(1+2\varepsilon)  ) \\
&\ge b_A + b_B +  (N_A -N(1-2\rho))  + 16 \varepsilon N \ge  b_A + b_B +  (N_A -N(1-2\rho))  + T. 
\end{align*}
Similarly we have $g_B > b_A +b_B + (N_B -N(1-2\rho)) + T$.

Using Lemma~\ref{lem:prog} we deduce that by time~$m(T)$ the number of good decodings Alice may have is bounded by $T+b_B$. From this point and on, every good decoding at Alice's side adds one vote for the correct leaf, making at least $g_A -  (T + b_B) > b_A +   (N_A -N(1-2\rho)) $ votes for that node by the end of the protocol, and at least $b_A+1$ votes until Alice reaches round $N(1-2\rho)$ when she gives her output.
On the other hand, any wrong output can get at most~$b_A$ votes, 
thus Alice outputs the correct leaf node at round~$N(1-2\rho)$. 
By a similar reasoning, Bob also outputs the correct leaf node when he reaches round~$N(1-2\rho)$. 
\end{proof}

\section{A coding scheme with a constant alphabet size}
\label{sec:constant}
Based on the protocol in Section \ref{sec:poly}, with some modifications, we obtain a protocol that has a constant size alphabet and a constant rate. To this end, we show how to encode each edge using varying-length encoding over a constant size alphabet. Although substantially more technically involved, this protocol is quite a straightforward extension of the protocol presented above. We thus defer the detailed analysis of this protocol to Appendix~\ref{sec:constant}.

\begin{theorem}
\label{thm:main2}
For any $\varepsilon>0$,
the simulation~$\pi'$ of Protocol~\ref{pro:pi'}
with 
$N = \lceil\frac{T}{\varepsilon^2}\rceil$, 
and a $(1-\varepsilon)$-edit distance tree code, 
solves $PJP(T)$ and 
is resilient to a $(1/18-\eps)$-fraction of edit corruptions.
\end{theorem}
\noindent Since $PJP(T)$ is complete for interactive communication, the above theorem proves  Theorem~\ref{thm:main}.


\bibliography{coding}

\begin{thebibliography}{10}

\bibitem{AGS13}
Shweta Agrawal, Ran Gelles, and Amit Sahai.
\newblock Adaptive protocols for interactive communication.
\newblock {\em arXiv preprint arXiv:1312.4182}, 2013.

\bibitem{AGHP92}
Noga Alon, Oded Goldreich, Johan H{\aa}stad, and Ren{\'e} Peralta.
\newblock Simple constructions of almost $k$-wise independent random variables.
\newblock {\em Random Structures \& Algorithms}, 3(3):289--304, 1992.

\bibitem{BK12}
Zvika Brakerski and Yael~Tauman Kalai.
\newblock Efficient interactive coding against adversarial noise.
\newblock In {\em Foundations of Computer Science, IEEE Annual Symposium on},
  pages 160--166. IEEE Computer Society, 2012.

\bibitem{BKN14}
Zvika Brakerski, Yael~Tauman Kalai, and Moni Naor.
\newblock Fast interactive coding against adversarial noise.
\newblock {\em J. ACM}, 61(6):35:1--35:30, December 2014.

\bibitem{BN13}
Zvika Brakerski and Moni Naor.
\newblock Fast algorithms for interactive coding.
\newblock In {\em SODA '13: Proceedings of the 24th Annual ACM-SIAM Symposium
  on Discrete Algorithms}, pages 443--456, 2013.

\bibitem{BNTTU14}
Gilles Brassard, Ashwin Nayak, Alain Tapp, Dave Touchette, and Falk Unger.
\newblock Noisy interactive quantum communication.
\newblock In {\em 55th {IEEE} Annual Symposium on Foundations of Computer
  Science, {FOCS} 2014, Philadelphia, PA, USA, October 18-21, 2014}, pages
  296--305, 2014.

\bibitem{BR14}
M.~Braverman and A.~Rao.
\newblock Toward coding for maximum errors in interactive communication.
\newblock {\em Information Theory, IEEE Transactions on}, 60(11):7248--7255,
  Nov 2014.

\bibitem{BE14}
Mark Braverman and Klim Efremenko.
\newblock List and unique coding for interactive communication in the presence
  of adversarial noise.
\newblock In {\em Foundations of Computer Science (FOCS), IEEE 55th Annual
  Symposium on}, pages 236--245, 2014.

\bibitem{BR11}
Mark Braverman and Anup Rao.
\newblock Towards coding for maximum errors in interactive communication.
\newblock In {\em STOC '11: Proceedings of the 43rd annual ACM symposium on
  Theory of Computing}, pages 159--166, New York, NY, USA, 2011. ACM.

\bibitem{EGH15}
Klim Efremenko, Ran Gelles, and Bernhard Haeupler.
\newblock Maximal noise in interactive communication over erasure channels and
  channels with feedback.
\newblock In {\em Proceedings of the 2015 Conference on Innovations in
  Theoretical Computer Science}, ITCS '15, pages 11--20, New York, NY, USA,
  2015. ACM.

\bibitem{FGOS13}
Matthew Franklin, Ran Gelles, Rafail Ostrovsky, and Leonard~J. Schulman.
\newblock Optimal coding for streaming authentication and interactive
  communication.
\newblock In Ran Canetti and Juan~A. Garay, editors, {\em CRYPTO '13}, volume
  8043 of {\em LNCS}, pages 258--276. Springer Berlin, 2013.

\bibitem{FGOS15}
Matthew Franklin, Ran Gelles, Rafail Ostrovsky, and Leonard~J. Schulman.
\newblock Optimal coding for streaming authentication and interactive
  communication.
\newblock {\em Information Theory, IEEE Transactions on}, 61(1):133--145, Jan
  2015.

\bibitem{GH15}
Ran Gelles and Bernhard Haeupler.
\newblock Capacity of interactive communication over erasure channels and
  channels with feedback.
\newblock In {\em Proceedings of the 26th Annual ACM-SIAM Symposium on Discrete
  Algorithms}, SODA '15, pages 1296--1311, 2015.

\bibitem{GMS11}
Ran Gelles, Ankur Moitra, and Amit Sahai.
\newblock Efficient and explicit coding for interactive communication.
\newblock In {\em Foundations of Computer Science, 2011 IEEE 52nd Annual
  Symposium on}, pages 768--777. IEEE Computer Society, 2011.

\bibitem{GMS14}
Ran Gelles, Ankur Moitra, and Amit Sahai.
\newblock Efficient coding for interactive communication.
\newblock {\em Information Theory, IEEE Transactions on}, 60(3):1899--1913,
  March 2014.

\bibitem{GH14}
Mohsen Ghaffari and Bernhard Haeupler.
\newblock {Optimal Error Rates for Interactive Coding~II: Efficiency and List
  Decoding}.
\newblock In {\em Foundations of Computer Science (FOCS), IEEE 55th Annual
  Symposium on}, pages 394--403, 2014.

\bibitem{GHS14}
Mohsen Ghaffari, Bernhard Haeupler, and Madhu Sudan.
\newblock Optimal error rates for interactive coding~{I}: Adaptivity and other
  settings.
\newblock In {\em STOC '14: Proceedings of the 46th Annual ACM Symposium on
  Theory of Computing}, pages 794--803, 2014.

\bibitem{haeupler14}
Bernhard Haeupler.
\newblock Interactive channel capacity revisited.
\newblock In {\em Foundations of Computer Science (FOCS), IEEE 55th Annual
  Symposium on}, pages 226--235, 2014.

\bibitem{justesen72}
J.~Justesen.
\newblock Class of constructive asymptotically good algebraic codes.
\newblock {\em Information Theory, IEEE Transactions on}, 18(5):652--656, Sep
  1972.

\bibitem{KR13}
Gillat Kol and Ran Raz.
\newblock Interactive channel capacity.
\newblock In {\em STOC '13: Proceedings of the 45th annual ACM Symposium on
  theory of computing}, pages 715--724, 2013.

\bibitem{KN97}
Eyal Kushilevitz and Noam Nisan.
\newblock {\em Communication Complexity}.
\newblock Cambridge University Press, 1997.

\bibitem{levenshtein66}
Vladimir~I. Levenshtein.
\newblock Binary codes capable of correcting deletions, insertions and
  reversals.
\newblock In {\em Soviet physics doklady}, volume~10, page 707, 1966.

\bibitem{ORS09}
R.~Ostrovsky, Y.~Rabani, and L.J. Schulman.
\newblock Error-correcting codes for automatic control.
\newblock {\em Information Theory, IEEE Transactions on}, 55(7):2931--2941,
  July 2009.

\bibitem{ORS05}
Rafail Ostrovsky, Yuval Rabani, and Leonard~J. Schulman.
\newblock Error-correcting codes for automatic control.
\newblock In {\em Foundations of Computer Science, Annual IEEE Symposium on},
  pages 309--316. IEEE Computer Society, 2005.

\bibitem{schulman92}
L.~J. Schulman.
\newblock Communication on noisy channels: a coding theorem for computation.
\newblock {\em Foundations of Computer Science, Annual IEEE Symposium on},
  pages 724--733, 1992.

\bibitem{SZ99}
L.~J. Schulman and D.~Zuckerman.
\newblock Asymptotically good codes correcting insertions, deletions, and
  transpositions.
\newblock {\em Information Theory, IEEE Transactions on}, 45(7):2552--2557, Nov
  1999.

\bibitem{schulman93}
Leonard~J. Schulman.
\newblock Deterministic coding for interactive communication.
\newblock In {\em STOC '93: Proceedings of the twenty-fifth annual ACM
  symposium on Theory of computing}, pages 747--756, New York, NY, USA, 1993.
  ACM.

\bibitem{schulman96}
Leonard~J. Schulman.
\newblock {Coding for interactive communication}.
\newblock {\em IEEE Trans. Inf. Theory}, 42(6):1745--1756, 1996.

\end{thebibliography}


\appendix

\section{More details about the noise model} 
\label{app:model}
Here we give some justifications that explains our modeling choices, 
and show it is weak enough to be a reasonable (i.e., non-trivial) model, yet strong enough to generalize other natural models.

As explained above, there are two main paradigms for distributed protocols in the asynchronous setting.
\begin{enumerate}
\item \textbf{Message-driven model:} In this model, each party wakes up and replies only when received a message. If several event occur ``simultaneously'', we assume a worst-case scenario in which the adversary determines the order of events.
We show that allowing arbitrary noise patterns make the protocol too strong. Either it halts in the middle, or the noise pattern limits the amount of interaction between the two parties.
\begin{enumerate}
\item If a deletion occurs (without being followed by a insertion, as in our model), then both parties clearly get stuck. More generally, it follows that at any point of the protocol, the number of insertions must exceed the number of deletions. We now show that this restriction still yields a too strong noise model which doesn't allow any resilient-protocols.
\item Assume the adversary is allowed to make up to $N/c$ insertions anywhere during the protocol. We show that this limits the protocol to performing $c+1$ interactions between Alice and Bob (where each ``interaction'' is sending a message of arbitrary length, which is depends on messages received in prior interactions). It then follows that one cannot obtain a constant-rate resilient protocol that withstands a constant fraction of noise $c$, as limiting the number of interactions may cause an exponential increase in the communication~\cite{KN97}.
\begin{claim}
Suppose Alice and Bob each send $N$ symbols in total and the adversary can make up to $N/c$ insertion/deletion errors. Then the adversary can make Alice and Bob have at most $c+1$ interactive rounds (alternations). 
\end{claim}
\begin{proof}
Without the loss of generality, let's assume Alice speaks first. The adversary makes the following attack:
\begin{enumerate}
\item After Alice sends Bob a symbol, the adversary inserts $N/c$ symbols from Alice to Bob. 
\item Bob receives all the $1+N/c$ symbols and replies all of them (recall, each incoming message triggers Bob to a single reply). The adversary guarantees that the order of event is such that first Bob answers on all messages and only then Alice receives them (these events are simultaneous, so Eve gets to decide their internal order). 
\item Alice receives all the $N/c+1$ bits and replies all of them, one by one. Again, Eve sets the internal order of events so that first Alice replies all messages before Bob start replying.  
\item Keep doing this until $N$ symbols are sent by both parties. 
\end{enumerate}
It's easy to see that the protocol has at most $c+1$ interactive rounds, where each interaction contains $N/c$ symbols which can be dependent on previous blocks.
\end{proof}
\end{enumerate}

\item \textbf{Clock-driven model:} In this model, each party has a clock and the party wakes up only when its clock ticks. The parties' clocks are not assumed to be synchronized or correlated in any way. We need to carefully describe what happens when the clock ticks is not alternating:
\begin{enumerate}
\item If a party wakes up and there's no incoming message, then we take a worst case assumption and allow the adversary to set the message that party sees.
\item If a party wakes up, and more than a single message was sent to him during that time, we will take a worst-case assumption that the party only sees the last incoming message.\footnote{We show a positive result for this setting, i.e., a resilient protocol. Thus it is better to take worst-case assumptions, as relaxing them may only get better resilience.}
\end{enumerate}
The following claim shows that the above model can be reduced to the model of edit-corruptions we use in this paper. 
\begin{claim}
The above clock-driven model can be reduced to a message-driven setting with edit corruptions. That is, any protocol in our model that is resilient to $\eps N$ edit corruptions, is also a resilient protocol in the clock-driven setting, which is resilient to  up to $\eps N$ ``non-alternating'' clock-ticks. A clock tick is considered not alternating if the previous clock tick belongs to the same same party.
\end{claim}
\begin{proof}
If there are two clock ticks made by the same party. Without loss of the generality, let's assume they are made by Alice. Denote these two clock ticks as $A(1)$ and $A(2)$, and assume that Bob's clock ticks after $A(2)$. From Bob's view, he sees only the message Alice sends at $A(2)$, so Alice's first message is deleted. When Alice wakes up at $A(2)$, which follows the time $A(2)$ (i.e., there is no other events in between), her incoming message queue is empty, so the adversary determines the symbol she sees; this is exactly an insertion. 
It follows that the above non-alternating clock-tick causes  exactly the same effect as of a deletion followed by an insertion. It is easy to verify that multiple non-alternating ticks have the same affect as multiple edit-corruptions.
\end{proof}
\end{enumerate}



\begin{figure}[p]
\centering
\begin{tikzpicture}[decoration=penciline]

\node (e) at (0.9,0.2) [token,fill=black,label=left:E] {};
\node (d) at (1.75,0.5) [token,fill=black,label=right:D] {};
\node (b) at (1.4,1) [token,fill=black,label=above right:B] {};
\node (a) at (0.95,1.85) [token,fill=black,label=left:A] {};

\draw[decorate,thick] (a) -- (b);
\draw[decorate,thick] (b) -- (d);
\draw[decorate,thick] (b) -- (e);

\draw (-0.7,0) --  (1,3)  -- (2.7,0);
\node (root) at (1,3.13) {};

\draw[decorate] (root) -- (a);
\end{tikzpicture}
\caption{An illustration of lambda structure}
\label{fig:lambda}
\end{figure}
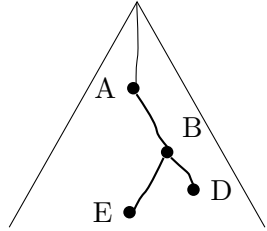

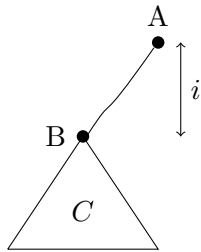
\begin{figure}[p]
\centering
\begin{tikzpicture}[decoration=penciline]

\draw (0,0) -- (1,1.5) -- (2,0) -- (0,0); 
\node at (1,0.5) {$C$};

\node (b) at (1,1.5) [token,fill=black,label=left:B]  {};
\node (a) at (2,2.75) [token,fill=black,label=above:A] {} ;
\draw[<->] (2.3,2.75)-- node[right]{$i$} (2.3,1.5);
\draw[decorate] (a)--(b);
\end{tikzpicture}
\caption{An illustration of the new tree $C'$ for defining the potential probability $P_i(C)$}
\label{fig:potential}
\end{figure}

\section{Missing proofs of Section~\ref{sec:edtc}} 
\label{app:treecode}

\begin{proof}[Proof of Lemnma \ref{lem:disuni}]
By the way we construct $\tilde S_n$ based on of $S_{n-1}$, and using the fact that for $i<n$, $\frac{|S_i|}{|\tilde{S}_i|} \geq c_2$, we have
\begin{eqnarray*}
\Pr_{C \uniform \tilde{S}_n} [ C(y)[1..j] = x]  &=& s^{-1} \cdot \Pr_{C \uniform S_{n-1}} [C(y)[2...j] = x[2..j]]\\
&\leq&  s^{-1}\cdot c_2^{-1} \cdot \Pr_{C \uniform \tilde{S}_{n-1}} [C(y)[2...j] = x[2..j]] \\
&\vdots& \\
&\leq& s^{-j+1} c_2^{1-j}\Pr_{C \uniform \tilde{S}_{n-j+1}} [ C(y)[j] = x[j]] \\
&=& s^{-j} c_2^{1-j}.
\end{eqnarray*}
\end{proof}

\begin{proof}[Proof of Lemnma \ref{lem:unidec}]

We prove the lemma by contradiction. Suppose there exist two messages $sm,sm' \in \cup_{i=1}^n C(\sum^n_{in})[1..i] $ such that both $SD(sm,rm)\leq \frac{\alpha}{2}$ and $SD(sm',rm)\leq \frac{\alpha}{2}$. We are going to show that this results in an $\alpha$-bad lambda. 

Consider the tree of the edit-distance tree code $C$. Let $D$ and $E$ be the tree nodes such that the path from root to $D$ denotes message $sm$ and the path from root to $E$ denotes message $sm'$. Let $B$ be the common ancestor of $D$ and $E$ in the tree.

Let $\tau$ and $\tau'$ be the string matching chosen in $SD(sm,rm)$ and $SD(sm',rm)$. Let $l = |\tau_1|$ and $l' = |\tau'_1|$. Choose $i, i'$ such that $del(\tau_1[i..l]) = BD$ and $del(\tau'
_1[i'..l']) = BE$. 
Next, consider the strings $del(\tau_2[i..l])$ and $del(\tau'_2[i'...l'])$. Both of them are suffixes of $rm$. Without loss of generality, let's assume $|del(\tau_2[i..l])| \leq |del(\tau'_2[i'..l'])|$. Therefore  $del(\tau_2[i..l])$ is also a suffix of $ del(\tau'_2[i'...l'])$ and there exists $j \leq i$ such that $del(\tau_2[j..l]) =del(\tau'_2[i'...l'])$. Now consider $del(\tau_1[j...l])$. Since $j \leq i$, $del(\tau_1[i...l])$ is a suffix of $del(\tau_1[j...l])$. So there is a node $A$ which is $B$'s ancestor or $B$ itself such that $AD$ corresponds to $del(\tau_1[j...l])$. 

Now we have
\begin{eqnarray*}
ED(AD, del(\tau_2[j..l])) &\leq& sc(\tau_1[j..l]) + sc(\tau_2[j...l]) \\
 &=& \frac{sc(\tau_1[j..l]) + sc(\tau_2[j..l])}{l - j + 1 -  sc(\tau_1[j..l])} \cdot (l - j + 1 -  sc(\tau_1[j..l])) \\
 &\leq& SD(sm,rm) \cdot |AD| \\ 
 &\leq& \frac{\alpha}{2} \cdot |AD|.\\
\end{eqnarray*}
Similarly, we have
\[
ED(BE, del(\tau'_2[i'..l'])) \leq \frac{\alpha}{2} \cdot |BE|.
\]
Since $del(\tau_2[j..l]) = del(\tau'_2[i'..l'])$ and by Fact \ref{fact:ed+}, we have
\begin{eqnarray*}
ED(AD, BE) &\leq& ED(AD, del(\tau_2[j..l])) + ED(BE, del(\tau'_2[i'..l'])) \\
&\leq& \frac{\alpha (|AD| + |BE|)}{2} \leq \alpha \cdot \max (|AD|, |BE|). 
\end{eqnarray*}
This shows that there's an $\alpha$-bad lambda in the tree code and therefore we get a contradiction.
\end{proof}


\section{Missing proofs of Section~\ref{sec:poly}} 
\label{app:poly}

\begin{proof}[Proof of Lemnma \ref{lem:gooddec}]
We prove the lemma for Alice, the lemma holds for Bob for the same reason. 
Consider the following procedure on the string matching $\tau_A = (\tau_1, \tau_2)$:
\begin{framed}
{\small
\begin{enumerate}[nosep]
\item Set $t = |\tau_2|$. 
\item While $t \neq 0$ do
\begin{enumerate}
\item If $\tau_2[t] = *$, then mark $t$ as bad and set $t \leftarrow t -1$. 
\item Elseif $\max_{i=1}^{t} \frac{sc(\tau_1[i..t]) + sc(\tau_2[i..t])}{t - i + 1 -  sc(\tau_1[i..t])} \geq \frac{\alpha}{2}$, then \\
\phantom{iiii} Let $t' = \arg\max_{i=1}^{t} \frac{sc(\tau_1[i..t]) + sc(\tau_2[i..t])}{t - i + 1 -  sc(\tau_1[i..t])}$. Mark $t'...t$ as bad, and set $t \leftarrow t' -1$. 
\item Else mark $t$ as good and set $t \leftarrow t - 1$. 
\end{enumerate}
\end{enumerate}
}
\end{framed}
First, we claim that the number of bad $t$'s is bounded by $ \frac{2}{\alpha}(sc(\tau_1) + sc(\tau_2)) +sc(\tau_1) $. The reason is the following:
In the above procedure, only step (a) and step (b) can mark some $t$'s as bad. 
Each time we mark $t'..t$ as bad (even if $t'=t$, e.g., by step (a)), then
we know that 
$
\frac{sc(\tau_1[t'..t]) + sc(\tau_2[t'..t])}{t - t' + 1 -  sc(\tau_1[t'..t])} \ge \frac{\alpha}{2}.
$
So the number of bad $t$'s among $t'..t$ is bounded by $\frac{2}{\alpha} (sc(\tau_1[t'..t]) + sc(\tau_2[t'..t])) + sc(\tau_1[t'..t])$.
Therefore, summing over all intervals $t'..t$ that were marked bad, 
the total amount of bad $t$'s  is bounded by
$
\tfrac{2}{\alpha}(sc(\tau_1) + sc(\tau_2)) + sc(\tau_1) .
$

Next, we claim that the number of good decodings is at least the number of good $t$'s. This implies that for Alice, the number of good 
decodings is $N_A + sc(\tau_2) - \frac{2}{\alpha}(sc(\tau_1) + sc(\tau_2)) -sc(\tau_1)$ since the number of good t's is equal to the length of $\tau_1$, $N_A+sc(\tau_2)$, minus the number of bad  $t$'s, $ \frac{2}{\alpha}(sc(\tau_1) + sc(\tau_2)) -sc(\tau_1)$. 
The claim holds  since each good~$t$ corresponds to a time that Alice receives some message from Bob and the suffix distance between the message Bob sent and the message Alice received is at most  
$\max_{i=1}^{t} \frac{sc(\tau_1[i..t]) + sc(\tau_2[i..t])}{t - i + 1 -  sc(\tau_1[i..t])} < \frac{\alpha}{2}$. 
By Lemma \ref{lem:unidec}, we know at that time the message Bob sent is the only one that has suffix distance less than $\frac{\alpha}{2}$ with respect to Alice's received word. So at that time, Alice decodes correctly. 
\end{proof}

\begin{proof}[Proof of Lemnma \ref{lem:prog}]

We will only prove the lemma for odd $i$'s. The lemma for even $i$'s will follow for the same reason. At time $m(i)$, Alice sends the $i$-th edge on $P$.  Then during times $m(i)+1,\ldots, m(i+1) -1$, there are no good decodings of Bob. Otherwise Bob would decode correctly and send the $i+1$-th edge on $P$. So during times $m(i)+1,...,m(i+1)-1$, all the good decodings are Alice's. Notice that for each good decoding that happens at Alice's side, at that very same time Alice receives a symbol that was sent by Bob (and not a symbol inserted by the adversary). 
Recall that Bob sends a symbol only after he performs a decoding, hence,
prior to each good decoding of Alice, there must be a decoding at Bob's side. 
However, there are no good decodings of Bob during this time period.  Also at time $m(i)$, the decoding is a decoding of Alice. So in the time period $m(i)+1,...,m(i+1) -1$, the number of good decodings of Alice  is at most the number of bad decodings of Bob. 
\end{proof}


\section{Detailed protocol and analysis for Theorem \ref{thm:main2}}
\label{sec:constant}

In this section we give a protocol~$\pi'$ that solves the $PJP(T)$ task, takes $N=O_\eps(T)$ rounds, 
and is resilient to $(1/18-\eps)$-fraction of edit corruptions. In contrast to Protocol~\ref{pro:pi} presented above, here $\pi'$ communicates symbols form a constant-size alphabet, and thus has communication complexity of
$CC(\pi')=O_\eps(T)$, that is, it has constant rate. By this we complete the proof of Theorem~\ref{thm:main}.
 \\
 \\
\textbf{The Outline.}
Generally speaking, $\pi'$ (Protocol~\ref{pro:pi'}) follows the same high-level ideas as Protocol~\ref{pro:pi},
except for encoding each edge $e$ transmitted in $\pi$ via a varying-length binary encoding.
More specifically, 
each edge $e = (n,s) \in [T]\times[4]$ is encoded as the binary description of $(i-n, s)$, where $i$ is the \emph{current round number} at the sender's side. The values $n$ and $s$ have the same meaning as in Protocol~\ref{pro:pi}: $n$ is a pointer to a previous round number in which the sender started to send an edge~$e'$, and $s$ indicates that the new edge $e$ is the $s$-th grandchild of the edge~$e'$. 
An empty edge is encoded as $(0,0)$, that is we assume $s$  take values in $\{0\} \cup [4]$.
Hence, the above binary description has length at most $\log(i-n) + 4$ bits. 
It can be shown that the amortized encoding length, is $O(1)$ and the proof is similar to the proof in~\cite{BR14}.

Several difficulties arise due the above varying-length encoding. 
First, since the channel's alphabet has a fixed size, it may take several rounds of communication to transmit a single edge. Furthermore, during these rounds where $e$ is being transmitted, new symbols are received from the other side. These symbols may cause the party to understand it needs to send a different edge $\tilde e$, instead of $e$ that is still in the process of being communicated. 

In the case explained above, the party will just add $\tilde e$ to a list of edges to be communicated. Recall that (some of) these edges may be added to the list due to adversarial noise. Moreover, such ``wrong'' edges may have a very large encoding (e.g., when the adversary causes a party to believe it needs to resend the first edge at the middle of the protocol; this has an encoding of length $\log (N)$). To prevent the adversary from delaying the progress of the protocol by these effects, we make the party send all the edges in the list \emph{in parallel}. That is, we cycle through the list of edges and send one bit from the encoding of any edge of the list. That way, long encodings do not delay the transmission of other edges. Furthermore, we attach to each edge a ``liveliness'' counter which indicates how likely it is for a specific edge to be part of the correct path of the underlying $PJP(T)$: each time a new symbol is received and the decoding of the incoming message implies some edge~$e$ should be sent, we increase the liveliness of that edge.
 This way, if an incorrect edge with a very long binary encdoing is added to the list, the noise must keep
 indicating this edge is needed in order to keep it ``alive'' in the list. 
\\ 
\\
\textbf{The Fine Details.}
We assume alphabet size of~$|\Sigma_{in}|=O(1/\eps^2)$.
Each party maintains a table (the EdgeTable), 
that stores all the edges this party is currently communicating. As mentioned above, the party cycles through the EdgeTable and sends one bit of every edge there. Thus, if the table holds $<1/\eps^2$ edges, a single round of communication suffices to send one bit of all the edges. Otherwise, several rounds of communication may be needed. 
Assume that the table holds $E$ edges, we name the process of sending each bit of all the $E$ edges, \emph{a cycle}. Note that a cycle takes $E/(1/\eps^2)$ rounds of communication; each such round is called \emph{a page}, and we say that a page is \emph{full} if it contains $1/\eps^2$ edges. For simplicity, we assume that new edges are added to the table only when a cycle is completed. Similarly, edges that were fully communicated are removed from the table at the end of a cycle.

%
The EdgeTable 
contains the  
following fields:
\begin{itemize}
\item \textbf{Edge:} contains the edge $e$ in the underlying $PJP(T)$ tree~$\T$.
\item \textbf{Edge Binary Description:} contains the varying-length binary description of $e$, as described above.
\item \textbf{Current Sent Length:} means how many bits of the binary description were already communicated. 
\item\textbf{Live Points:} A number describing the liveliness of the edge. When $e$ is added to the table it gets $1/\eps$ live points. At any future round where $e$ is to be added to the table, if it is already in the table it's live points increase by $1/\eps$. Every round we communicate a bit of $e$ its live points decrease by~1. If the live points of some edge reach $0$ at the end of a cycle, the edge is removed from the table.
\end{itemize}

Let us now formally define the process of a cycle: transmitting one bit from each edge in EdgeTable.
In fact, along with one bit per edge, we also send some meta-data, which is described
in Procedure~\ref{proc:cycle} below. Each symbol in the our encoding
$\Sigma_{in}$ needs to hold $4$ bits for up to~$1/\eps^2$ edges, thus
it is clear that $|\Sigma_{in}|=O(1/\eps^2)$ suffices.
We note that the information sent at each cycle suffices to exactly recover the   EdgeTable table
at the other side.
\begin{Procedure}[H]
\small
repeat $\lceil E/(1/\eps^2)\rceil $ times ($E$ is the number of edges in EdgeTable):

The next symbol to communicate consists of the following information: 
\begin{enumerate}
\item One bit to indicate whether this page is the last page in the cycle or not.
\item For each $e$ of the next $1/\eps^2$ edges in $E$ include:
\begin{enumerate}
\item the next bit of $e$'s ``Edge Binary Description'' to be transmitted
\item One bit to indicate whether $e$ has ``Live Points'' = 0. 
\item One bit to indicate whether $e$ has ``Current Sent Length'' = the length of its ``Edge Binary Description''. 
\item One bit to indicate whether $e$ was added to the table at the end of the last cycle.
\end{enumerate}
\end{enumerate}
\caption{Cycle(EdgeTable)}
\label{proc:cycle}
\end{Procedure}

We now describe the UpdateTable($e$) procedure. This procedure is called after every cycle, adds the edge $e$ to the cycle, and removes edges that either were fully transmitted or their ``Live points'' has reached 0.
The parties add edges to their EdgeTable using the 
update process UpdateTable($e$) described in Procedure~\ref{proc:updateTable}. 
\begin{Procedure}[H]
\caption{UpdateTable($e$)}
\label{proc:updateTable}\small
\begin{enumerate}
\item Add $e$:
\begin{enumerate}
\item If $e$ is not an empty edge and $e$ is not in the table, then insert $e$ to the table. Compute the ``Edge Binary Description''. Set ``Current Sent Length'' = 0 and ``Live Points'' {= $\frac{1}{\varepsilon}$}. \looseness=-1
\item If $e$ is not an empty edge and $e$ is in the table, increase its ``Live Points'' by $\frac{1}{\varepsilon}$.
\end{enumerate}
\item Maintain liveliness, and remove dead edges:
\begin{enumerate}
\item Decrease ``Live Points'' of each edge in the table by 1. Increase ``Current Sent Length'' of each edge in the table by 1. 
\item Remove all the edges with ``Live Points'' = 0 or ``Current Sent Length'' = the length of its ``Edge Binary Description''. 
\end{enumerate}
\end{enumerate}
\end{Procedure}

We are now ready to describe the Protocol~$\pi'$,  given in Protocol~\ref{pro:pi'}. Similar as in Section \ref{sec:poly}, we say that $d_A$ is valid if it is the encoding of a set of edges.
\begin{Protocol}[H]
\caption{The protocol $\pi'$}
\label{pro:pi'}
\small
Let $\T$ be given by $PJP(T)$. Recall that Alice's input is $X$. Assume the parties share some fixed $\alpha$-edit-distance tree code $C:\Sigma_{in}^{N} \rightarrow \Sigma_{out}^{N}$. 

\medskip
Initialize $i = 0$. \\
Repeat for $N=\lceil T/\eps^2\rceil$ times: 
\begin{enumerate}[itemsep=1pt]
\item $ i \leftarrow i  + 1$. 
\item Receive a symbol $r_A[i]$ from the other party. (For Alice, if $i = 1$, skip this step)
\item Find $d_A \in \Sigma_{in}^*$ which minimizes $SD(d_A,r_A[1...i])$. 
\item If $E(d_A) \cup X$ has a unique path from the root in $PJP(T)$, do the following. Here $E(d_A)$ is a set of edges indicated by $d_A$, if $d_A$ is not a valid message, $E(d_A) = \emptyset$. 
\begin{enumerate}
\item If this path already reach the leaf node $v$, then $s(v) \leftarrow s(v) + 1$. 
\item Let $e$ be the deepest edge on the the unique path from root. If $e \not \in X$ , or $e$ is not an edge in the first or second level of $\T$ or $e$'s grandparent has been sent, set $e$ to be an empty edge.
\end{enumerate} 
\item If $E(d_A) \cup X$ does not have a unique path from the root in $PJP(T)$, set $e$ to be an empty edge.
\item If a cycle is completed (or if $i=1$), call UpdateTable($e$). 
\item $s_A[i] \leftarrow $ the next page of Cycle(EdgeTable).
\item Send $C(s_A[1...i])[i]$ to Bob. 
\item If $i = N (1- 2\rho)$, output the leaf node $v$ with the largest $s(v)$. 
\end{enumerate}
\end{Protocol}

We now analyze the coding scheme of Protocol~\ref{pro:pi'}, and prove our main theorem (Theorem~\ref{thm:main}, obtained via Theorem~\ref{thm:main2}). 
First, we claim that Protocol~\ref{pro:pi'} has a constant rate. 
Indeed, it communicates $2N \log |\Sigma_{out}|$ bits throughout. We use a tree code with
$|\Sigma_{in}|=O_\eps(1)$ and $\alpha=1-\eps$ and thus by 
Theorem~\ref{thm:edtc} we get that $|\Sigma_{out}|=O_\eps(1)$ as well. 
Hence, the total communication of $\pi$ is given by 
$$CC(\pi') = 2N\log |\Sigma_{out}| = T/\eps^2 \cdot (1/\eps^2)^{O(1/\eps)}= O_\eps(T).$$

The correctness analysis is quite similar to the one of Section~\ref{sec:poly}. 
As above, 
let $N_A$ and $N_B$ be the counter~$i$ of Alice and Bob when one of them reaches the end of the protocol. Let $\tau_A = (\tau_1, \tau_2)$ be the string matching between $s_B[1..N_B]$ and $r_A[1..N_A]$ that is consistent with the protocol. Let $\tau_B = (\tau_3, \tau_4)$ be the string matching between $s_A[1..N_A]$ and $r_B[1..N_B]$. Recall that we use $sc(\tau)$ to denote the number of $*$'s in the string. By definition,  $sc(\tau_1) + sc(\tau_3)  \leq 2\rho N$ and $sc(\tau_2) + sc(\tau_4) \leq 2\rho N$.

Let's first prove the following lemma which shows that the case where the number of edges in the table exceeds $\frac{1}{\varepsilon^2}$ is very rare. 
\begin{lemma}
\label{lem:page}
Alice sends at most $\varepsilon \cdot  N_A$ full pages, and Bob sends at most $\varepsilon \cdot  N_B$ full pages (recall that a full page is a page that contains $\frac{1}{\varepsilon^2}$~edges).
\end{lemma}
\begin{proof}
We prove the lemma for Alice, and the same holds for Bob. Notice that the sum of ``Live Points'' of all edges in the table is always non-negative. And in each round of Alice, if the UpdateTable procedure is called, it increases the sum of ``Live Points'' by at most $\frac{1}{\varepsilon}$.  Whenever Alice sends a page of $\frac{1}{\varepsilon^2}$ edges, the sum of ``Live Points'' is decreased by at least $\frac{1}{\varepsilon^2}$. Alice has $N_A$ rounds in total. Let $p_A$ be the number of full pages sent by Alice. By a simple counting argument, we have
\[
N_A \cdot \frac{1}{\varepsilon} - p_A \cdot \frac{1}{\varepsilon^2} \geq 0.
\]
Therefore $p_A \leq \varepsilon\cdot  N_A$. 
\end{proof}

We again define good decoding, which is slightly different from the one of Section~\ref{sec:poly}---we only care about decodings that happen at the end of a cycle. 
\begin{definition}[Good Decoding]
\label{def:gooddec2}
When a party decodes a message, we say that it is a good decoding if the followings hold:
\begin{enumerate}[itemsep=2pt]
\item The decoding outputs the entire message sent so far by the other side ($d_A = s_B[1..i]$ or $d_B = s_A[1..i]$)
\item The symbol just received was not inserted by the adversary. 
\item The party has finished a cycles (i.e., UpdateEdge(e) is called at that round.)
\end{enumerate}
Otherwise we call it a bad decoding.
\end{definition}

\begin{lemma}
\label{lem:gooddec2}
Alice has at least 
$N_A(1-\varepsilon) + (1-\frac2\alpha)sc(\tau_2) -(1+\frac2\alpha) sc(\tau_1)$ 
good decodings. Bob has at least 
$N_B(1-\varepsilon) + (1-\frac2\alpha)sc(\tau_4) -(1+\frac2\alpha)sc(\tau_3)$ good decodings. 
\end{lemma}
\begin{proof}
We prove the lemma for Alice; the lemma for Bob holds for the same reason. By Lemma~\ref{lem:gooddec}, the number of decodings of Alice which satisfy the first two constraints in Definition~\ref{def:gooddec2} is at least $N_A +(1-\frac2\alpha)sc(\tau_2) -(1+\frac2\alpha) sc(\tau_1)$. By Lemma \ref{lem:page}, the number of decodings of Alice which don't satisfy the third constraint in Definition \ref{def:gooddec2} is at most $\varepsilon\cdot  N_A$. So Alice has at least $N_A + (1-\frac2\alpha)sc(\tau_2) -(1+\frac2\alpha) sc(\tau_1) - \varepsilon \cdot N_A$ good decodings. 
\end{proof}

As in the analysis of Protocol~\ref{pro:pi}, we now sort all the good and bad decodings of both Alice and Bob by the time these decodings happen. Since in total we have $N_A + N_B$ decodings, we can assume these decodings happen at times $t=1,2,...,N_A + N_B$. 
Again, we assume that for each decoding, the sending procedure right after it happens at the same time as the decoding.  Let $e_A(t)$ be the set of edges Alice has fully communicated up to time~$t$ and $e_B(t)$ be the set of edges Bob has fully communicated up to~$t$. Let $P$ be the correct path of length $T$ defined by the underlying~$PJP(T)$. 
Let~$l(t)$  be the length of the longest path starting from root using edges from~$P \cap (e_A(t) \cup e_B(t))$. Basically $l(t)$ measures how much progress Alice and Bob have made. Let's also define $m(i)$ to be the first time $t$ such that $l(t) \geq i$.  

The following lemma is similar to Lemma~\ref{lem:prog} and shows that many bad decodings are needed to slow down the progress. For notation convenience, we need the following definition.
\begin{definition}
\label{def:region}
Let $g_A[t_1,t_2]$ and $g_B[t_1,t_2]$ be the number of good decodings of Alice and the number of good decodings of Bob during times $t_1,...,t_2$. Let $b_A[t_1,t_2]$ and $b_B[t_1,t_2]$ be the number of bad decodings of Alice and  the number of bad decodings of Bob during times $t_1,...,t_2$.
\end{definition}

\begin{lemma}
\label{lem:prog2}
For $i = 0,...,T-1$, let $l_{i+1}$ be the length of the binary description of the $(i+1)$-th edge on~$P$ during the first transmission where this edge is fully communicated.  During time $m(i)+1,...,m(i+1)$, the following is true. 
\begin{enumerate}
\item If $i$ is odd, then 
\begin{align*}
g_B[m(i)+1,m(i+1)] &\leq l_{i+1} + \varepsilon\cdot  b_B[m(i)+1,m(i+1)], \text{ and }\\ g_A[m(i)+1,m(i+1)] &\leq  l_{i+1} + (1+\varepsilon)\cdot  b_B[m(i)+1,m(i+1)].
\end{align*}
\item If $i$ is even, then 
\begin{align*}
g_A[m(i)+1,m(i+1)] &\leq l_{i+1} + \varepsilon \cdot b_A[m(i)+1,m(i+1)], \text{ and } \\g_B[m(i)+1,m(i+1)] &\leq  l_{i+1} + (1+\varepsilon)\cdot b_A[m(i)+1,m(i+1)].
\end{align*}
\end{enumerate}
\end{lemma}
\begin{proof}
We will only prove the lemma for odd $i$'s; the case for even $i$'s follow for the same reason.
First, for any time period $t_1,...,t_2$, it holds that $g_A[t_1,t_2] \leq g_B[t_1,t_2] + b_B[t_1,t_2]$. This follows since every good decoding of Alice stems from a symbol that was actually sent by Bob (rather than inserted by Eve), which implies a decoding at Bob's side (either bad or good).

Since we assume $i$ is odd, 
,at time $m(i)$, Alice finishes sending the $i$-th edge on path $P$. From the above claim we have $g_A[m(i)+1,m(i+1)] \leq g_B[m(i)+1,m(i+1)] + b_B[m(i)+1,m(i+1)]$. Then to prove Lemma \ref{lem:prog2}, it suffices to show that $g_B[m(i)+1,m(i+1)] \leq l_{i+1} + \varepsilon \cdot b_B[m(i)+1,m(i+1)]$ because the upper bound on $g_B[m(i)+1,m(i+1)]$ will directly give us the upper bound on  $g_A[m(i)+1,m(i+1)]$.

Let's divide time period $m(i)+1,...,m(i+1)$ into three different types of time periods and bound the number of good decodings of Bob. Let $e_{i+1}$ be the $(i+1)$-th edge of~$P$.
\begin{enumerate}[itemsep=2pt]
\item Intervals $[t_1,t_2]$ in which $e_{i+1}$ is not in Bob's EdgeTable. In this interval $g_B[t_1,t_2] = 0$ because if Bob has a good decoding, he adds the ($i+1$)-th edge of~$P$ into his EdgesTable.

\item Intervals $[t_1,t_2]$ where $e_{i+1}$ was inserted into the Bob's EdgeTable at~$t_1$, but removed at~$t_2$ before Bob has finished to fully send its encoding.
In this time period, the ``Live Points'' of $e_{i+1}$ is increased for at least $g_B[t_1,t_2]$ times, and each time it is increased by $\frac{1}{\varepsilon}$. 
At the same time the ``Live Points'' of  $e_{i+1}$ decreases by 1 for at most $b_B[t_1,t_2]$ times, until they reach $0$ at time $t_2$.
We thus have $g_B[t_1,t_2] \leq \varepsilon \cdot b_B[t_1,t_2]$.

\item Interval $[t_1,t_2]$ where $e_{i+1}$  was inserted into the Bob's EdgeTable at $t_1$, and finally this edge's sending procedure is finished at time $t_2$.
It is clear that $t_2 = m(i+1)$, and  $g_B[t_1,t_2] \leq l_{i+1}$ since assuming $e_{i+1}$ wasn't removed from the table,  after $l_{i+1}$ times of sending all the pages in the table, Bob will have finished sending the encoding of $e_{i+1}$. Recall that a good decoding implies Bob completed sending the last page in his table. 
\end{enumerate}
To sum up, we know that in case 1 and 2, we have $g_B[t_1,t_2] \leq \varepsilon \cdot b_B[t_1,t_2]$. And case 3 will happen at most once. Therefore we have $g_B[m(i)+1,m(i+1)] \leq l_{i+1} + \varepsilon \cdot b_B[m(i)+1,m(i+1)]$ which completes the proof.
\end{proof}

In the following lemma, 
we bound the sum of the length of encodings of $P$'s edges, regardless of the progress of the protocol. 

\begin{lemma}
\label{lem:desc}
Given any instance of $\pi'$, that succeeds to fully communicate the first
$T'\le T$ edged in~$P$, 
for $i=1,...,T'$ let $l_i$ be the length of the binary encoding of 
$i$-th edge on $P$ 
as communicated in that instance. For $i>T'$, define $l_i=0$.
Then,
\[
\sum_{i=1}^T l_i \leq T \left(\log\left(\frac{N_A+N_B}{T}\right) + 4\right)
\]
\end{lemma}
\begin{proof}
For any edge $l_i$ let $(\delta_i,s_i)$ be the encoding that was used in that instance of $\pi'$. By definition, $\delta_i$~is the difference between the round $n_i$ where we start sending the encoding of $l_i$ and round $n_{i-2}$, where the first bit of $l_{i-2}$ was communicated.
Clearly, 
for any $i > 2$ we have $l_i \leq \log(n_i - n_{i-2}) + 4$, while for the first two edges $i = 1,2$ we have  $l_i \leq \log(n_i -0) + 4$. 

Then we have
$$n_1 + \sum_{\substack{3 \leq i \leq T',\\ i \text{ odd}}}(n_i - n_{i-2}) \leq N_A
\quad\text{ and } \quad
n_2 + \sum_{\substack{4 \leq i \leq T',\\ i \text{ even}}}(n_i - n_{i-2}) \leq N_B,$$ 
and by the concavity of logarithm, we have
\[
\frac1{T'}\sum_{i=1}^{T'} l_i 
\leq 
	\log 	\left (
	  	\frac1{T'}\Big( n_1 +n_2 + \sum_{3\leq i \leq T} (n_i - n_{i-2}) \Big)
	        \right)
	+ 4
\leq \log\left(\frac{N_A+N_B}{T'}\right) + 4.
\]
Finally, note that $x\log(N/x)$ is monotonic increasing in $x\in(0,N/e)$,
thus the claim holds for any~$T$ as long as $T'\le T\le 2N/e$. 
\end{proof}

With the above lemmas we can complete the proof of the main theorem.
\begin{theorem}[Restate of Theorem \ref{thm:main2}]
For any $\varepsilon>0$,
the simulation~$\pi'$ of Protocol~\ref{pro:pi'}
with 
$N = \lceil\frac{T}{\varepsilon^2}\rceil$, 
and a $(1-\varepsilon)$-edit distance tree code, 
solves $PJP(T)$ and 
is resilient to a $(1/18-\eps)$-fraction of edit corruptions.
\end{theorem}

\begin{proof}
For notation convenience, let $g_A  = g_A[1,N_A+N_B]$, $g_B  = g_B[1,N_A+N_B]$, $b_A  = b_A[1,N_A+N_B]$ and $b_B  = b_B[1,N_A+N_B]$. 
Also let $\rho=1/18-\eps$ and $\alpha=1-\varepsilon$.
The proof is quite similar to the one of Theorem \ref{thm:main1} with some small difference in the parameters, specifically,
we show that both Alice and Bob have a large number of good decodings (compared to their bad decodings) which implies they both output the correct value.

\begin{claim}
\label{clm:main2}
\[
g_A -  (\sum_{i=1}^T l_i) \ge  (N_A - (1-2\rho) N) -(1+\varepsilon)(b_A + b_B). 
\]
\end{claim}
\begin{proof}
Recall  that 
both $N_A$ and $N_B$ are in the range~$[N(1-2\rho), N]$
because there are at most $2\rho N$ insertion/deletion errors. By Lemma \ref{lem:gooddec2}, we have 
\begin{align*}
&b_A + b_B  \\
& \le N_A -g_A  + N_B-g_B \\
& \le 
	\varepsilon N_A +\tfrac{2}{\alpha}(sc(\tau_1) + sc(\tau_2)) + sc(\tau_1) - sc(\tau_2) + 
	\varepsilon N_B + \tfrac{2}{\alpha}(sc(\tau_3) + sc(\tau_4)) + sc(\tau_3) -sc(\tau_4) \\
& \le \frac{8 \rho N}{\alpha}   + 2\varepsilon N
\end{align*}
Using Lemma~\ref{lem:desc} to bound $\sum_i l_i$, and using the fact that $\rho+ \varepsilon \leq \frac{1}{18}$, and that $N\ge T/\varepsilon^2$, we have 
\begin{align*} \displaybreak[2]
&g_A - (1+\varepsilon)(b_A + b_B) - (\sum_{i=1}^T l_i) -(N_A - (1-2\rho) N)\\*
&\geq N_A - b_A - (1+\varepsilon)(b_A + b_B) - T (\log(\frac{N_A+N_B}{T}) + 4) - (N_A - (1-2\rho) N) \\
&\geq N_A -  (2+\varepsilon)(b_A + b_B) - T (\log(\frac{N_A+N_B}{T}) + 4) - (N_A - (1-2\rho) N) \\
&\geq (1-2\rho) N - (2+\varepsilon)( \frac{8 \rho N}{\alpha}  + 2\varepsilon N ) - T (\log(\frac{2N}{T})+4)\\
&\ge 11\varepsilon N  - T (\log(\frac{2N}{T})+4) \\
&\ge T( \frac{11}{\varepsilon} - 4  -\log(\frac{2}{\varepsilon^2})) \\
&> 0
\end{align*}
\end{proof}
Similarly, we have $g_B - (1+\varepsilon)(b_A + b_B) - (\sum_{i=1}^T l_i) - (N_B - (1-2\rho) N) > 0$. 

With the above claim and Lemma \ref{lem:prog2}, it follows that both Alice and Bob will finish sending all the edges on $P$. 
Also by Lemma \ref{lem:prog2}, we have 
\[
g_A[1,m(T)] \leq \varepsilon \cdot b_A[1,m(T)]  + (1+\varepsilon)b_B[1,m(T)] + (\sum_{i=1}^T l_i ).
\]
Therefore, we have
\[
g_A[m(T)+1,N_A+N_B] =  g_A - g_A[1,m(T)] > b_A + (N_A - (1-2\rho) N).
\]
As Alice could have at most $(N_A - (1-2\rho) N)$ good decodings after she gives her output, the number of good decodings in the period between $m(T)$ and the point where Alice gives an output, is larger than $b_A$, suggesting the correct leaf will get at least that many votes.
On the other hand, any bad leaf will get at most $b_A$ votes, thus Alice outputs the correct leaf. A similar reasoning applies for Bob.
\end{proof}


\section{Potent edit-distance tree code}
\label{sec:potent}
In this section we show that it is simple to construct a relaxed notion of \EDTC which is very useful for most applications, namely potent edit-distance tree codes (\PEDTC). 
This follows by extending the techniques of Gelles, Moitra and Sahai for constructing a relaxed notion of standard tree codes~\cite{GMS11}, to the edit-distance case.
\begin{definition}(Bad Interval)
For a given $d$-ary prefix code of depth $n$  
let (A,B,D,E) be an $\alpha$-bad lambda (Definition \ref{def:lambda}). 
Let $h$ be the depth of $A$ and $\ell=\max(|AD|,|AE|)$. 
We say that the  interval $[h, h+\ell]$ is \emph{$\alpha$-bad}.
\end{definition}


\begin{definition}[$(\delta,\alpha)$-\PEDTC]
\label{def:PEDTC}
An \emph{$(\delta,\alpha)$-bad tree}
is a prefix code of depth $n$
that has a path $Q$ for which the union of all bad intervals along $Q$ (i.e., all bad intervals for which either the point $E$ or $D$ of the bad-alpha inducing that bad interval is a node on~$Q$)
has a total length of at least~$\delta n$.

If the tree is not an {$(\delta,\alpha)$-bad tree}, then we will call it 
a $(\delta,\alpha)$-potent edit-distance tree-code, or $(\delta,\alpha)$-\PEDTC for short.
\end{definition}

In other words, in a \PEDTC, for any given root-to-leaf path $Q$ in the tree, 
if we sum up all the nodes that belong to some bad lambda, their number would be less than~$\delta n$. 
Next we show that a tree whose labels are randomly chosen is a \PEDTC with high probability.

\begin{proposition}
For any constants $\delta,\alpha \in (0,1)$ and for any $d \geq 2$,
and for any $|\Sigma_{out}| \ge (32d^3)^{\frac{8}{(1-\alpha)\delta}}$, 
a $d$-ary tree whose labels are independently and uniformly  selected from~$\Sigma_{out}$ is a  ($\delta,\alpha$)-\PEDTC with probability~$\ge 1 - 2^{-n}$.
\end{proposition}
\begin{proof}
Assume we construct a tree by randomly picking each label uniformly from~$\Sigma_{out}$. 
Assume the obtained tree code is not ($\delta,\alpha$)-potent. Then, it means there exist a path $Q$ and a set of nodes $\{x_1, \ldots, x_d\}$ along it, so that for each $x_i$ there exists points $\lambda_{x_i}=(A,B,D,E)_{x_i}$ where $x_i=E$ or $x_i=D$ such that 
$\lambda_{x_i}$
is a $\alpha$-bad lambda that induces the interval~$l_i$, and that 
$|\bigcup_i l_i| \ge\delta n$.

 The following technical lemma  says that there exists a subset of $\{x_1, \ldots, x_k\}$ whose nodes induce bad intervals of length at least $\delta n/2$ and the intervals are \emph{disjoint}.
Also, note that there are at most $2^n \times 2^n = 4^{n}$
ways to distribute these disjoint
intervals along the path~$Q$. This is because, in order to figure out the configuration of all intervals, we only have to determine which nodes are in some interval and which nodes are ends of some interval.

\begin{lemma}[\cite{schulman96}]\label{lem:intervals}
Let $\ell_1,\ell_2,\ldots,\ell_k$ be intervals on $\mathbb{N}$, and their union has length~$X$.
Then there exists a set of indices $I\subseteq \{1, 2, \ldots, k\}$ such that
the intervals indexed by $I$ are disjoint, and their
total length is at least $X/2$. That is, for any $i,j\in I$, $|\ell_i \cap \ell_j|=0$, and
$
\sum_{i\in I}\lvert\ell_i\rvert \ge X/2
$.
\end{lemma}
\noindent A proof is given in~\cite{schulman96}.

Now let's first consider the probability that some interval~$l_i$ is bad when the path~$Q$ is given. Let's first see which lambda ($A,B,D,E)$ can make the interval bad. We count these lambdas by figuring out points $A,B,D,E$ one by one. As $Q$ and $l_i$ are given, the point $A$ is fixed. 
Next let's figure out the positions of $D$ and $E$. We know that one of $D$ and $E$ should be on $Q$ and $\max (|AE|, |AD|) = |l_i|$. It follows that
there are at most $2|l_i| \cdot d^{|l_i|+1}$ ways to pick $E$ and $D$:
there are $\le |l_i|$ ways to pick the point which lies on~$Q$, and $\le d^{|l_i|+1}$ ways to pick the point which is not on~$Q$; another factor of~2 is for choosing whether $D$ or $E$ is the one that lies on~$Q$.
Also note that
each such choice determines the position of $B$ along~$Q$. 
So in total, there are at most $2|l_i| \cdot d^{|l_i|+1}$ lambdas that could induce the bad interval~$l_i$ on~$Q$. 

Given any lambda $\lambda_{x_i}$ 
Lemma~\ref{lem:edpr} bounds that the probability that $\lambda_{x_i}$ is $\alpha$-bad by 
\[
|\Sigma_{out}|^{-\frac{(1-\alpha)\max(|AD|,|BE|)}{2}} \leq  |\Sigma_{out}|^{-\frac{(1-\alpha)|l_i|}{4}}.
\] 
So the probability that $l_i$ on $Q$ is $\alpha$-bad is at most $2|l_i|d^{|l_i|+1}\cdot |\Sigma_{out}|^{-\frac{(1-\alpha)|l_i|}{4}}$. 

Using Lemma~\ref{lem:intervals} above, there exists a set $I$ so that $\{\lambda_{x_i}\}_{i\in I}$ induce disjoint intervals with a total length of at least~$\sum_{i\in I} |l_i| \ge \delta n/2$.
Since the intervals are disjoint, their probabilities to occur
are independent, and
the probability that a specific pattern of intervals happens is bounded by their product.
By setting $|\Sigma_{out}|\ge (32d^3)^{\frac{8}{(1-\alpha)\delta}}$
we can bound the probability 
that a \PEDTC is $(\delta,\alpha)$-bad by
\begin{align*}
\Pr[\text{ \PEDTC is }(\delta,\alpha)\text{-bad }]  
&\leq \sum_Q \sum_{\substack{ l_1,...,l_k: \\ \sum_{i=1}^k |l_i| \geq \delta n /2}} \prod_{i=1}^k Pr[l_i \text{ is a bad interval}]\\
&\leq \sum_Q \sum_{\substack{ l_1,...,l_k: \\ \sum_{i=1}^k |l_i| \geq \delta n /2} }\prod_{i=1}^k 2|l_i|d^{|l_i|+1}\cdot |\Sigma_{out}|^{-\frac{(1-\alpha)|l_i|}{4}}\\
&\leq d^n \cdot 4^n  \cdot 2^n\cdot 2^n \cdot  d^{2n} \cdot |\Sigma_{out}|^{-\frac{(1-\alpha)\delta n}{8}} \\
&\le (16d^3)^n \cdot (32d^3)^{-n} \\ 
&= 2^{-n}. \qedhere
\end{align*}
\end{proof}
Similar to~\cite{GMS11,GMS14} we can further derandomize the above construction, by making the random choice of each label coming out of a small-biased sample space~(e.g.,~\cite{AGHP92}). This immediately leads to an efficient (randomized) construction of \PEDTC with succinct description that succeeds with overwhelming probability. We omit the details and refer the reader to~\cite{GMS11,GMS14}.


\section{A coding scheme with a constant alphabet size and \PEDTC }
\label{sec:PEDTCpro}

In this section, we show that our scheme work the same when the edit-distance tree code is replaced with a potent edit-distance tree. Specifically,
we prove that  Protocol~\ref{pro:pi'} still solves the pointer jumping problem given a \PEDTC. 
 
First, we show an analogue of Lemma~\ref{lem:unidec} used for decoding the tree, to the case of potent trees.
\begin{lemma}
\label{lem:unidec2}
Let $C$: $\Sigma_{in}^n \rightarrow \Sigma_{out}^n$ be a $(\delta,\alpha)$-\PEDTC, and let $rm \in \Sigma_{out}^m$ ($m$ can be different from~$n$). If there exist $sm \in \cup_{i=1}^n C(\Sigma_{in}^n)[1..i]$ such that $SD(sm,rm) \leq \frac{\alpha}{2}$ and the end node of $sm$ on the tree is not in any bad interval of any path $Q$ that contains $sm$, then $sm$ will be the unique message $\in \cup_{i=1}^n C(\Sigma_{in}^n)[1..i]$ such that $SD(sm,rm) \leq \frac{\alpha}{2}$.
\end{lemma}

\begin{proof}
We prove this lemma by contradiction. Suppose there exist two messages $sm,sm' \in \cup_{i=1}^n C(\sum^n_{in})[1..i] $ such that both $SD(sm,rm)\leq \frac{\alpha}{2}$ and $SD(sm',rm)\leq \frac{\alpha}{2}$, and also the end node of $sm$ on the tree is not in any bad interval of any path $Q$ that contains $sm$. By the same argument in Lemma \ref{lem:unidec}, we know there is an $\alpha$-bad lambda in $C$. Let this $\alpha$-bad lambda to be $(A,B,D,E)$, then either $D$ or $E$ would be the end point of $sm$. Therefore, the end point of $sm$ would be inside some bad interval of some path $Q$. Now we get a contradiction.
\end{proof}

We can now prove the main theorem of this section, namely, the existence of a coding scheme that assumes $\PEDTC$.
\begin{theorem}
\label{thm:main3}
For any $\varepsilon>0$,
the simulation~$\pi'$ of Protocol~\ref{pro:pi'}
with 
$N = \lceil\frac{T}{\varepsilon^2}\rceil$, 
and an $(\varepsilon, 1-\varepsilon)$-\PEDTC instead of an edit distance tree code, 
solves $PJP(T)$ and 
is resilient to a $(1/18-\eps)$-fraction of edit corruptions.
\end{theorem}

\begin{proof}
The proof is very similar to the proof of Theorem \ref{thm:main2}. We are going to use exactly the same notations like Definition \ref{def:gooddec2} and Definition \ref{def:region}. And it is easy to check that Lemma \ref{lem:page}, Lemma \ref{lem:prog2} and Lemma \ref{lem:desc} still hold for Protocol \ref{pro:pi'} with \PEDTC. The only difference is that we need to make a slight change in Lemma \ref{lem:gooddec2}. We prove the following lemma as the analogy of Lemma \ref{lem:gooddec2}.

\begin{lemma}
\label{lem:gooddec3}
Alice has at least 
$N_A(1-\varepsilon) + (1-\frac2\alpha)sc(\tau_2) -(1+\frac2\alpha) sc(\tau_1) - \varepsilon N_B$ 
good decodings. Bob has at least 
$N_B(1-\varepsilon) + (1-\frac2\alpha)sc(\tau_4) -(1+\frac2\alpha)sc(\tau_3) - \varepsilon N_A$ good decodings. 
\end{lemma}
\begin{proof}
We prove the lemma for Alice; the lemma for Bob holds for the same reason. Compare this lemma with Lemma \ref{lem:gooddec2}, it suffices to show that the number of new bad decodings of Alice caused by changing edit distance tree code to PEDTC is at most $\varepsilon N_B$. For an instance of running Protocol \ref{pro:pi'} with \PEDTC, let $Q$ be the path of $s_B[1..N_B]$(all the symbols sent by Bob) on the tree of \PEDTC $C$. By comparing Lemma \ref{lem:unidec} and Lemma \ref{lem:unidec2} we know that the corresponding sending message of each new bad decoding has end point in some bad interval of $Q$. We also know that each new bad decoding satisfies the constraint that the symbol just received in this decoding was not inserted by the adversary. So the corresponding sending messages of new bad decodings have different end points on the tree of \PEDTC. By the definition of $(\varepsilon,1-\varepsilon)$-\PEDTC, we know that there are at most $\varepsilon N_B$ nodes on $Q$ which are in bad intervals. Therefore, the number of new bad decodings is at most $\varepsilon N_B$. 
\end{proof}

Next we are going to prove an analogue of Claim \ref{clm:main2} for Protocol \ref{pro:pi'} with \PEDTC. 
\begin{claim}
\label{clm:main3}
\[
g_A -  (\sum_{i=1}^T l_i) \ge  (N_A - (1-2\rho) N) -(1+\varepsilon)(b_A + b_B) 
\]
\end{claim}
\begin{proof}
By Lemma \ref{lem:gooddec3}, we have
\begin{align*}
&b_A + b_B  \\
& \le N_A -g_A  + N_B-g_B \\
& \le 
	2\varepsilon N_A +\tfrac{2}{\alpha}(sc(\tau_1) + sc(\tau_2)) + sc(\tau_1) - sc(\tau_2) + \\
&~~~~~	2\varepsilon N_B + \tfrac{2}{\alpha}(sc(\tau_3) + sc(\tau_4)) + sc(\tau_3) -sc(\tau_4) \\
& \le \frac{8 \rho N}{\alpha}   + 4\varepsilon N
\end{align*}
Using Lemma~\ref{lem:desc} to bound $\sum_i l_i$, and using the fact that $\rho+ \varepsilon \leq \frac{1}{18}$, and that $N\ge T/\varepsilon^2$, we have 
\begin{align*}
&g_A - (1+\varepsilon)(b_A + b_B) - (\sum_{i=1}^T l_i) -(N_A - (1-2\rho) N)\\
&\geq N_A - b_A - (1+\varepsilon)(b_A + b_B) - T (\log(\frac{N_A+N_B}{T}) + 4) - (N_A - (1-2\rho) N) \\ \displaybreak[0]
&\geq N_A -  (2+\varepsilon)(b_A + b_B) - T (\log(\frac{N_A+N_B}{T}) + 4) - (N_A - (1-2\rho) N) \\ \displaybreak[0]
&\geq (1-2\rho) N - (2+\varepsilon)( \frac{8 \rho N}{\alpha}  + 4\varepsilon N ) - T (\log(\frac{2N}{T})+4)\\ \displaybreak[0]
&\ge 6\varepsilon N  - T (\log(\frac{2N}{T})+4) \\
&\ge T( \frac{6}{\varepsilon} - 4  -\log(\frac{2}{\varepsilon^2})) \\
&> 0
\end{align*}
\end{proof}
After proving Claim \ref{clm:main3}, in order to finish the proof, we just need to exactly follow the argument after Claim \ref{clm:main2} in Theorem \ref{thm:main2}.
\end{proof}


\section{Upper bound on the noise fraction}
\label{sec:upperbound}

In this section we show that no protocol can tolerate edit-corruption noise fraction
of $\rho=1/6$ or more, assuming that the parties are required to be correct after receiving $(1-2\rho)N$ symbols. Intuitively, since the effective protocol is shorter, 
the noise bound of~$1/4$ (that stems from the ability to perform substitutions) becomes higher. Indeed, $1/4 \cdot (1-2\rho)N$ corruptions out of a total communication of~$2N$ amounts to a fraction of errors of
\[
\rho = \frac{1/4 \cdot (1-2\rho)N}{2N} \quad \Longrightarrow \quad \rho = \frac16. 
\]

We remind that the requirement to give an output comes from Eve's possibility to delete~$2\rho N$ at one party, making this party receive at most $(1-2\rho)N$ symbols throughout the protocol. Below we give a formal attack that demonstrate the bound of~$1/6$, under this stringent requirement. 
\begin{theorem}
Let $\pi$ be an interactive protocol for $PJP(T)$ that is resilient to 
a $\rho$-fraction of edit corruptions, and assume $|\pi|=N$.
If both parties are required to output the correct output at round $(1-2\rho)N$, 
and assuming an adversarial edit-corruption rate of~$\rho=1/6$, 
then no protocol outputs the correct output with probability $>1/2$.
\end{theorem}
\begin{proof}
Assume the parties run an instance of $\pi$ on inputs $(x,y)=(0,0)$.
Eve performs the following attack (wlog, attacking Alice): 
For the first $N/3$ symbols sent by Alice, Eve deletes those symbols and inserts back to Alice symbols that simulate Bob on input $y=1$. Then (at round $N/3$ and beyond), Eve does nothing.

Recall that Alice must output the correct value after receiving $(1-1/3)N=2N/3$ symbols. However her view at this time is indistinguishable from an instance of $\pi$ on inputs $(x,y)=(0,1)$ in which Eve corrupts Alice rounds $[N/3, 2N/3]$ by deleting the symbols Alice is sending in these rounds, and inserting Bob's symbols assuming $y=0$, and assuming the $N/3$-th symbol sent by Alice is the first received by Bob. This implies that Alice cannot determine Bob's input at round\footnote{We mention that it may be possible for Alice to output the correct answer by round~$N$.} $2N/3$ under this attack, which proves the claim. 
\end{proof}

If we don't require the the parties to output the correct value already at round $(1-2\rho)N$, then it is possible to show an upper bound of $\rho=1/4$ on the fraction of noise, similar to the case of interactive protocols over standard noisy channels~\cite{BR14}.


\end{document}